\documentclass[aps,pra,amssymb,superscriptaddress,10pt]{revtex4-2}
\usepackage{braket}
\usepackage{graphicx}
\usepackage{amsthm}
\usepackage{amsmath}
\usepackage{amssymb}
\usepackage{color}
\usepackage{soul}
\usepackage{url}
\usepackage{mathptmx}
\usepackage{times}
\usepackage{txfonts}

\newcommand{\Bg}[1]{\mathcal{B}_{#1}}
\newcommand{\Bgd}[1]{\mathcal{B}_{#1}^\dagger}
\newcommand{\Br}[1]{\mathcal{B}_{#1}}

\begin{document}

\title{Entanglement-Inducing Quantum Markov Processes}
\author{Jonas Fransson\href{https://orcid.org/0000-0002-9217-2218}{\includegraphics[scale=0.05]{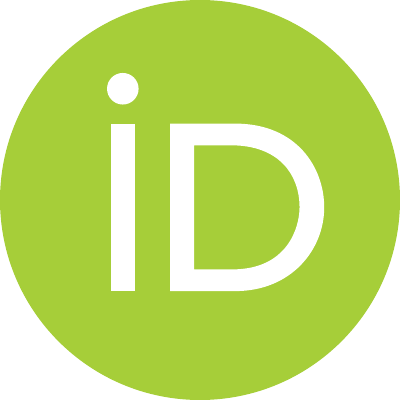}}}
\affiliation{Department of Physics and Astronomy, Materials Theory, Uppsala University, 751 37 Uppsala, Sweden}
\email{jonas.fransson@physics.uu.se}
\author{Artur P. Sowa\href{https://orcid.org/0000-0001-7833-2716}{\includegraphics[scale=0.05]{orcidid.pdf}}}
\affiliation{Department of Mathematics and Statistics, University of Saskatchewan, Saskatchewan S7N 5E6, Canada}
\email{sowa@math.usask.ca}


\maketitle
\newtheorem{definition}{Definition}
\newtheorem{theorem}{Theorem}
\newtheorem{proposition}{Proposition}
\newtheorem{lemma}{Lemma}
\newtheorem{corollary}{Corollary}
\newtheorem{algorithm}{Algorithm}
\newtheorem{conjecture}{Conjecture}

\begin{center}

 Abstract

 \end{center}
We introduce a new model for a system of interacting bosons placed in an array of sites. At its core is a nonlinear, nonlocal evolution equation, which we have dubbed the Schr\"{o}dinger--Dirichlet equation. The construction is closely related to the Bose--Hubbard model and to a specific type of generalized bosons. In contrast to conventional mean-field closures, the resulting nonlinear dynamics need not preserve product structure and can generate entanglement from initially separable states. The relevant methods of analysis are based on harmonic analysis for the multiplicative group of positive rationals.

\vspace{.2cm}

  \vspace{.2cm}


\section{Introduction}

This work is the latest step in a natural progression of concepts pertaining to the mathematical treatment of an infinite array of boson sites. It was demonstrated in \cite{SF22} that such structures are naturally amenable to analysis via a number-theoretic framework in which the Fock space is represented as $\ell_2(\mathbb{N})$ and the bosonic creation and annihilation operators act multiplicatively on states. This setting was further developed in \cite{SF25}, leading to the introduction of nonlocal coherent states. Further properties of such coherent states were subsequently investigated in \cite{FSS26}. Besides proving fruitful for the analysis of bosons, this framework also naturally leads to the introduction of generalized bosons of a specific type.

Generalized bosons are foundational to the present work. In essence, the generalized boson creation and annihilation operators retain the structural features of the canonical bosonic operators: they modify the local particle numbers of a state. However, they additionally temper amplitudes, which renders them bounded operators, see Section  \ref{Subsection_gen_boson}. In many respects this makes the analysis of various structures technically simpler than in the canonical bosonic setting. Remarkably, the number operators associated with bosons and generalized bosons become statistically indistinguishable in the limit of a large grand canonical ensemble, see Subsection \ref{Subsection_gen_boson}. This conclusion is justified by the seminal Erd\H{o}s--Kac theorem from probabilistic number theory, \cite{Erdos-Kac}.

Within the framework of generalized bosons, every state may be naturally expressed as arising from the vacuum state through the action of a suitable creation operator. Repeated application of such state-dependent operators gives rise to a form of nonlinearity that we explore here. The underlying mathematical structure is that of Dirichlet series. Briefly, the construction begins with the transform
\[
\ket\psi = \sum_{n} x_n \ket n \mapsto \Bgd{\psi} = \sum_{n} x_n \hat{b}_n^\dagger
\]
which assigns to a Fock-space state a generalized-bosonic creation operator. Remarkably, the resulting nonlinear operation
\[
\ket\psi \mapsto \Bg{\psi} \Bgd{\psi} \ket\psi
\]
need not preserve product structure and can generate entanglement from initially separable states. This stands in contrast to conventional mean-field approximations, which restrict the many-body dynamics to a prescribed product-state structure.

Pursuing this line of reasoning has led us to propose a nonlinear model for the description of certain regimes of interacting bosons on an infinite array of boson sites. The model may be viewed either as describing an array of generalized boson sites or as an effective model for an array of canonical boson sites. The proposed dynamics are based on a universal evolution equation with a cubic nonlinear term, (\ref{SDE_main_Four}). The model retains the usual Bose--Hubbard site-local energy term while reworking the interaction term. We have dubbed the resulting equation the Schr\"{o}dinger--Dirichlet equation (SDE).

The equation may be viewed as describing a quantum Markov process. Apart from its pure-state formulation, it also admits a generalization to the evolution of mixed states, (\ref{Boson-Schrod-Mixed}). Unlike the well-known dissipative Markovian master equations, the SDE is not dissipative in character. Rather than driving decoherence or relaxation toward classical mixtures, it is natural to expect the nonlinear dynamics to generate entanglement.

The nonlinearity arises from a nonlocal self-interaction energy. Due to the specific structure of this interaction, it is natural to analyze the model within the framework of harmonic analysis on the multiplicative group of positive rationals $\mathbb{Q}_+$ and its dual, the infinite-dimensional torus $\mathbb{T}^\omega$, an approach first introduced in this context in \cite{SF22}. We summarize the basic features of this apparatus in Section~\ref{Appendix_FT_Q_plus}.

The Fourier-dual reformulation of the SDE highlights both formal similarities and essential differences between the present model and models based on the Gross--Pitaevskii equation or Ginzburg--Landau theory, now classical in the boson lattice literature; see, e.g., \cite{Polkov}, \cite{Grass}. In those models, the bosonic creation and annihilation operators are replaced by classical c-numbers, leading to a cubic nonlinearity that is local in the spatial variables. In contrast, the principal nonlinear term of the SDE is nonlocal.

The use of generalized Fourier transform also leads to discretization of the SDE. This in turn enables the resolution of the eigenvalue problem in the single particle subspace, see Subsection \ref{Subsect_finite_sols}. Furthermore, the universal equation gives rise to infinitely many systems of Hamiltonian evolution equations, simply by assuming that the state is supported on a fixed and finite number of sites and has a fixed number of particles, see example in Subsection \ref{Subsection_SDE}.   

The deepest question is the existence of the DSE ground state, which reduces to the problem of existence of the global energy-minimizing state in the fixed number of particles subspace. This part of analysis is presented in Section \ref{Section_minimizer}. We have not been able to settle that question directly. However, we give a satisfactory answer in the case of a tempered energy functional, see Theorem \ref{theorem_minimizer}. 

In summary, we believe the present approach is fruitful and nuanced enough to give insight into bosonic and generalized bosonic systems that goes beyond the standard theory. The main enabling features of this approach are: replacing unbounded operators by bounded ones, analytic formulation. It needs to be pointed out that the analysis in question takes place on a space that is not a differential manifold, and in this sense is best characterized as being exo-geometric.      

\section{Prequel: classical entropy-producing Markov processes with evolutionary features}
\label{subsection_Markov}

Consider probability distributions $P,Q : \mathbb{N} \rightarrow [0,1] $, which we represent as vector columns  $P = [p_1, p_2, p_3, \ldots ]^T$ and $Q = [q_1, q_2, q_3, \ldots ]^T$, where $\sum_{n} p_n = \sum_n q_n = 1$. Recall the Dirichlet convolution of sequences $P$ and $Q$, say $R = P\star Q$, is a sequence $R = [r_1, r_2, r_3, \ldots ]^T$ defined via
\begin{equation}\label{Dirichlet_PQ}
  r_n = \sum_{d|n} p_d q_{n/d},\quad \mbox{ for } n = 1,2,3, \ldots
\end{equation} 
where the sum is over all divisors $d$ of $n$. As is well known,  $R$ is, again, a probability distribution. A simple way to reach this conclusion is via the following interpretation. Namely, consider an infinite matrix constructed with the entries of $P$, namely: 
\begin{equation}\label{thePhat}
\hat{P} =
 \left[\begin{array}{cccccccccc}
 &p_1 &\cdot &\cdot &\cdot &\cdot &\cdot &\cdot &\cdot &\ldots\\
  &p_2  &p_1 &\cdot &\cdot &\cdot &\cdot &\cdot &\cdot &\ldots\\
  &p_3 &\cdot &p_1 &\cdot &\cdot &\cdot &\cdot &\cdot &\ldots\\
  &p_4 &p_2 &\cdot &p_1 &\cdot &\cdot &\cdot &\cdot &\ldots\\
  &p_5 &\cdot &\cdot &\cdot &p_1 &\cdot &\cdot &\cdot &\ldots\\
  &p_6 &p_3 &p_2 &\cdot &\cdot &p_1 &\cdot &\cdot &\ldots\\
  &p_7 &\cdot &\cdot &\cdot &\cdot &\cdot &p_1 &\cdot &\ldots\\
  &p_8 &p_4 &\cdot &p_2  &\cdot &\cdot &\cdot &p_1&\ldots\\
 &\vdots &\vdots &\vdots &\vdots &\vdots  &\vdots &\vdots &\vdots &\ddots
 \end{array}\right] ,
\end{equation}
where single dots stand for zeros. 
It is then easily seen that $R = \hat{P} Q$. However, $\hat P$ is a Markov matrix as the sum of entries in each column is $1$. Therefore, $R$ is a probability distribution. Note that the convolution is commutative and we could have alternatively said $R = \hat{Q} P$ with an obvious interpretation of $\hat Q$. 

Recall that the entropy of $P$ is defined as $E(P) = -\sum p_k\log p_k$ where, say, we chose $\log$ to denote the natural logarithm. It is understood that $0 \log 0 = \lim_{x\rightarrow 0+} x \log x = 0$. Also, we will say that $P$ is supported at one point if there is $k$, such that $p_k=1$. The following observation plays important role in what is to follow: 
\begin{lemma}\label{Lem_Entropy}
For $R = P \star Q$, we have $E(R) \geq \max\{E(P),E(Q)\}$ with equality if and only if $P$ or $Q$, or both, are supported at one point.
\end{lemma}
\begin{proof}
  We will use a generalization of well-known fact, \cite{Ash}. Let $[x_1,x_2,x_3,\ldots ]$ be a sequence of positive numbers satisfying $\sum_k x_k \leq 1$. For  $p_k>0$, one has $\log (x_k/p_k) \leq x_k/p_k -1$  with equality if and only if $x_k = p_k$. It follows that 
  \[
  \sum_k p_k \log\frac{x_k}{p_k} \leq \sum_k (x_k - p_k) \leq 1-1= 0 ,\quad \mbox{ with equality iff } p_k = x_k \mbox{ for all } k.
  \]
  Thus, 
  \begin{equation}\label{basic_estim}
     -\sum_k p_k\log p_k \leq  -\sum_k p_k\log x_k,\quad \mbox{ with equality iff } p_k = x_k \mbox{ for all } k.
  \end{equation}
  Next, observe the following identity
  \[
   E(R) = -\sum r_n \log r_n = - q_1 \sum_j p_j \log r_j - q_2 \sum_j p_j \log r_{2j} - q_3 \sum_j p_j \log r_{3j} - \ldots
  \]
  Note that whenever $r_{kj} = 0$, then $q_kp_j = 0$, in which case $q_kp_j\log r_{kj} = q_kp_j \log p_j =0$. Together with (\ref{basic_estim}), this implies
    \[
  - q_k \sum_j p_j \log r_{kj} \geq  - q_k \sum_j p_j \log p_j.
  \] 
Thus,
  \[
   E(R) \geq - \sum_k q_k \, \sum_j p_j\log p_j   = E(P).
\]
  Note that equality holds only if for a certain $k$ we have $r_{kj} = p_j$ for all $j$, which implies $q_k =1$, i.e. $Q$ is supported at one point. Since the roles of $P$ and $Q$ can be switched, this proves the lemma. 
\end{proof}

\subsection{A discrete Markov process based on the Dirichlet convolution}

We are interested in Markov processes obtained via the Dirichlet convolution. First, we define the discrete version:
\begin{equation}\label{Markov_k}
  P^{(1)} = P, \quad P^{(k+1)} = P\star P^{(k)}, \quad \mbox{ for } k = 2,3,4,\ldots 
\end{equation} 
Observe that if $P$ is a stationary distribution under this process, then $p_1 =1$. Therefore, $P = [ 1,0,0,\ldots]^T$ is the only stationary point of the process. Also, excluding the case when $P$ is supported at one point, the entropy is increasing at every step. Furthermore, the evolution of the first $N$ terms of the distribution does not depend on any other entries of the sequence. Indeed,
 \begin{equation}\label{Markov_k_eq}
 \left[
   \begin{array}{c}
     p_1^{(k+1)} \\
     p_2^{(k+1)} \\
     p_3^{(k+1)} \\
     \vdots \\
     p_N^{(k+1)} \\
   \end{array}
 \right]
 =
 \left[\begin{array}{cccccccccc}
 &p_1 &\cdot &\cdot &\cdot &\cdot &\cdot &\cdot &\cdot\\
  &p_2  &p_1 &\cdot &\cdot &\cdot &\cdot &\cdot &\cdot \\
  &p_3 &\cdot &p_1 &\cdot &\cdot &\cdot &\cdot &\cdot \\
    &p_4 &p_2 &\cdot &p_1 &\cdot &\cdot &\cdot &\cdot \\
 &\vdots &\vdots &\vdots &\vdots &\vdots  &\vdots &\vdots &\vdots \\
  &p_N & p_{N/2} & p_{N/3} & \ldots  &\ldots &\ldots &\ldots &p_1\\
 \end{array}\right] 
 \left[
   \begin{array}{c}
     p_1^{(k)} \\
     p_2^{(k)} \\
     p_3^{(k)} \\
     \vdots \\
     p_N^{(k)} \\
   \end{array}
 \right]
\end{equation}
where we adopt a convention that $p_{N/k} = 0$ whenever $k$ is not a divisor of $N$. Let $S_N^{(k)} = p_1^{(k)} + \ldots p_N^{(k)}$. Note that the sum of the entries in every column of the matrix above does not exceed $S_N = S_N^{(1)}$. Therefore, 
\begin{equation}\label{S_N_bound}
  S_N^{(k)} \leq (S_N)^k.
\end{equation}
Thus, whenever $S_N <1$ we have $ S_N^{(k)} \rightarrow 0$ when $k$ tends to infinity. In particular, if the initial distribution $P$ has infinite support, then $ S_N^{(k)} \rightarrow 0$ for all $N$, i.e., in successive iterations the weight is transferred toward increasingly ``remote locations". Consider the complementary case, when the initial distribution is finitely supported. Let $n>1$ be the largest index such that $p_n>0$. In such a case, the remotest locations with nonzero weight are consecutively $p^{(2)}_{n^2}>0$, $p^{(3)}_{n^4}>0$, $p^{(4)}_{n^8}>0$, etc. Thus, the consecutive distributions remain finitely supported but the weight is again transferred toward more and more remote locations.  

\subsection{A continuous Markov process based on the Dirichlet convolution}

We now introduce a continuous-time Markov process, analogous to (\ref{Markov_k}), namely:
\begin{equation}\label{Markov_t}
\frac{d}{dt} P(t) = P(t)\star P(t) - P(t), \quad \mbox{ for } t\in [0,\infty) \mbox{ with the initial condition }   P(0).
\end{equation}
Since we know of no reference to this process, we will start the discussion with a few elementary observations. First, since $P(t)\star P(t)$ is a probability distribution, it follows that
\begin{equation}\label{sum_P}
  \frac{d}{dt} \sum_{n} p_n(t) = 0.
\end{equation}
Second, as in the discrete case, the evolution of the first $N$ terms in the series $P(t)$ does not depend on the rest of the coefficients. Specifically,
 \begin{equation}\label{Markov_t_eq}
 \frac{d}{dt}\,\,\left[
   \begin{array}{c}
     p_1(t) \\
     p_2(t) \\
     p_3(t) \\
     p_4(t) \\
     \vdots \\
     p_N(t) \\
   \end{array}
 \right]
 =
 \left[\begin{array}{cccccccccc}
 &p_1(t)-1 &\cdot &\cdot &\cdot &\cdot &\cdot &\cdot &\cdot\\
  &p_2(t)  &p_1(t)-1 &\cdot &\cdot &\cdot &\cdot &\cdot &\cdot \\
  &p_3(t) &\cdot &p_1(t)-1 &\cdot &\cdot &\cdot &\cdot &\cdot \\
    &p_4(t) &p_2(t) &\cdot &p_1(t)-1 &\cdot &\cdot &\cdot &\cdot \\
 &\vdots &\vdots &\vdots &\vdots &\vdots  &\vdots &\vdots &\vdots \\
  &p_N(t) & p_{N/2}(t) & p_{N/3}(t) & \ldots  &\ldots &\ldots &\ldots &p_1(t)-1\\
 \end{array}\right] 
 \left[
   \begin{array}{c}
     p_1(t) \\
     p_2(t) \\
     p_3(t) \\
    p_4(t) \\
     \vdots \\
     p_N(t) \\
   \end{array}
 \right]
\end{equation}
This is a system of differential equations with polynomial entries on the right-hand side. The existence of solutions for finite time and their uniqueness follows from the standard theory. Furthermore,
\begin{equation}\label{p1_t}
  \frac{d}{dt}p_1 = p_1^2-p_1.
\end{equation}
Note this is the logistic equation with reversed time. We easily obtain the conclusion that either  $p_1(t) \equiv 1$, or $p_1(t) = 1/(1+ce^t)$, where  $c = -1+ 1/p_1(0)$. From now on, we are interested only in the latter case, i.e., we assume $p_1(0)<1$. We do have 
\[
\frac{d}{dt} p_k(t) = \sum_{d|k} p_d(t) p_{k/d}(t) - p_k(t) \geq -p_k(t)\quad \mbox{ for } k = 1, 2, \ldots N.
\] 
It follows that
\begin{equation}\label{p_positive}
   p_k(t)\geq p_k(0)\,e^{-t}. 
\end{equation}
Thus, all terms remain non-negative which together with (\ref{sum_P}) implies that the right hand side of (\ref{Markov_t_eq}) is given by polynomials on the cube $[0,1]^{\times N}$. Therefore, solutions exist for all times. Since $N$ is arbitrary, $P(t): \mathbb{N}\rightarrow [0,1]$ is a uniquely defined probability distribution  that exists for all times .  
 
Furthermore, reasoning as in the case of the discrete process, we observe that the partial sum $S_N(t) = p_1(t) + \ldots + p_N(t)$ satisfies
\begin{equation}\label{bound_S_N_t}
  \frac{d}{dt} S_N(t) \leq S_N(t)^2 - S_N(t).
\end{equation} 
Had there been equality, this would have been the logistic equation with reversed time. Recall that we assume $p_1(0)<1$. It follows that for any $t>0$ we have $S_N(t) <1$, i.e., after arbitrarily short time, $P(t)$ has infinite support. Indeed, if $p_n(0)>0$ then either $p_{n^2}$ is already positive, or its derivative is positive. This means that in arbitrarily short time $p_{n^k}$ will turn positive, if they have not been so from the start.   
Therefore, fixing an arbitrary $\epsilon >0$, we have
\begin{equation} \label{bound_S_N_t_expl}
S_N(t) \leq \frac{1}{1+ce^t}\quad  \mbox{ for }\,\, t \geq\epsilon >0, \quad \mbox{ where } c = \frac{1}{S_N(\epsilon)} -1. 
\end{equation}

Next, we demonstrate that the entropy of $P(t)$ is increasing. For simplicity, we consider the generic case when $p_n(0) > 0$ for all $n$. We have the following:
\begin{theorem}\label{Theor_Entropy}
  Let $P = P(t)$ satisfy (\ref{Markov_t}) with the generic initial condition, i.e. $p_n(0) >0 $ for all $n$. Then, the entropy of $P$ is strictly increasing at all times, i.e., 
  \[
  \frac{d}{dt} E(P) >0.
  \] 
\end{theorem} 
\begin{proof}
  Let $R(t) = P(t)\star P(t)$. Note that in light of (\ref{p_positive}), $p(t)>0$ for all times. Next, we have
  \begin{eqnarray*}
    \frac{d}{dt} E(P) = -  \frac{d}{dt} \sum_k p_k \log p_k  &=& -  \sum_k \left(  \frac{d}{dt} p_k\right) \log p_k - \sum_k   
    \frac{p_k }{p_k} \, \frac{d}{dt} p_k\\
     &=& -  \sum_k \left(  \frac{d}{dt} p_k\right) \log p_k = -  \sum_k \left(r_k - p_k \right) \log p_k\\
        &=&  -  \sum_k r_k \log p_k  - E(P) \geq  E(R) - E(P) >0,
  \end{eqnarray*}
   where the last two inequalities follow from (\ref{basic_estim}) and from Lemma \ref{Lem_Entropy}. This completes the proof. 
\end{proof}

In summary, generically, the process (\ref{Markov_t}) constantly produces entropy and shifts weight from less complex states $\ket n$ (with $n$ less composite) to more complex states $\ket n$ (with $n$ more composite). 

\subsection{Explicit representation of the Markov processes via harmonic analysis on the group of positive rationals}

We will derive a closed-form solution of (\ref{Markov_t}). It helps to understand the relation between the two processes (\ref{Markov_k}) and (\ref{Markov_t}). To this end, we utilize a mollified version of the Fourier transform for the group of fractions, see Subsection \ref{Appendix_FT_Q_plus}. Namely, for a $\sigma >0$, we send $P$ to a function $\hat{P}_\sigma: \mathbb{T}^\omega \rightarrow \mathbb{C}$, via
\begin{equation}\label{mol_FT}
  P(t) \longrightarrow \hat{P}_\sigma(\vec{\mu}, t) : = \sum_{n} p_n(t)\, n^{-\sigma} e^{2\pi i \vec{n}\cdot \vec{\mu}}.
\end{equation} 
A direct calculation shows that (\ref{Markov_t}) is equivalent to the equation
\begin{equation}\label{Markov_t_FT}
  \frac{d}{dt} \hat{P}_\sigma(\vec{\mu}, t) = \hat{P}_\sigma(\vec{\mu}, t)^2 - \hat{P}_\sigma(\vec{\mu}, t).
\end{equation}
The point is that $\hat{P}_\sigma(\vec{\mu}, t)$ can be manipulated as a scalar rather than a sequence. Next, we assume $p_1(0)>0$ which ensures that for a sufficiently large $\sigma>0$, we have 
\[
|\hat{P}_\sigma(\vec{\mu}, 0)| \geq p_1(0) - \sum_{n=2}^{\infty} p_n(0)\, n^{-\sigma} \geq \epsilon >0 \mbox{ for all } \vec{\mu} \in \mathbb{T}^\omega.
\]  
Of course, we automatically have 
\[
|\hat{P}_\sigma(\vec{\mu}, 0)| \leq \sum_{n=1}^{\infty} p_n(0)\, n^{-\sigma} < 1 \mbox{ for all } \vec{\mu} \in \mathbb{T}^\omega, \quad (\mbox{ provided } p_1(0)<1 ). 
\] 
With these assumptions, we readily obtain
\begin{equation}\label{Markov_t_sol}
  \hat{P}_\sigma(\vec{\mu}, t) = \frac{1}{1+ \hat{C}(\vec{\mu})e^t}, \quad \mbox{ where } \,\, \hat{C}(\vec{\mu}) = \frac{1}{\hat{P}_\sigma(\vec{\mu}, 0)} -1 .
\end{equation}
The formula gives an immediate corollary: If the initial condition $\hat{P}_\sigma(\vec{\mu}, 0)$ effectively depends only on finitely many variables $\mu_p$, then so does the solution for all times. This is equivalent to the statement that $p_n(t)$ can be nonzero only for those $n$ whose prime factors belong to the distinguished finite set of primes. 

It is interesting to observe that using formula (\ref{Markov_t_sol}), one may represent $ \hat{P}_\sigma(\vec{\mu}, t)$  as a series, namely
 \begin{equation}\label{Markov_t_series}
  \hat{P}_\sigma(\vec{\mu}, t) = e^{-t}\, \sum_{n=1}^{\infty} (1-e^{-t})^{n-1} \hat{P}_\sigma(\vec{\mu}, 0)^n.
\end{equation}
This shows a connection with the discrete process (\ref{Markov_k}) as, indeed, $\hat{P}_\sigma(\vec{\mu}, 0)^n$ is the Fourier transform of $P^{(n)}$. 
\vspace{.5cm}

\noindent \emph{Example.} Suppose the initial condition $P$ is supported at one point, $p_q(0) =1$ where $q$ is prime. We have 
$\hat{P}_\sigma(\vec{\mu}, 0) = q^{-\sigma} \exp{2\pi i \mu_q}$. Formula (\ref{Markov_t_series}) assumes the simple form
\[
  \hat{P}_\sigma(\vec{\mu}, t) = e^{-t}\, \sum_{k=1}^{\infty} (1-e^{-t})^{k-1} q^{-k\sigma}\,e^{2\pi i k \mu_q} .
\] 
This means that the only nonzero weights of $P(t)$ are $p_{q^k}$, $k =1, 2, \ldots$, and
\begin{equation}\label{P_t_example}
 p_{q^k}(t) = e^{-t}\, (1-e^{-t})^{k-1}.
\end{equation}
In this (non-generic) case, the entropy can be calculated explicitly:
\[
E(P(t))
=
t-(e^t-1)\log(1-e^{-t})
\sim t+1,
\qquad t\to\infty.
\]
Since the expectation value of the distribution is $e^t$, its characteristic
scale grows exponentially, while its entropy grows logarithmically with that
scale. This is consistent with the fact that the geometric distribution
maximizes entropy among probability distributions on $\mathbb{N}$ with a
prescribed expectation value.

\section{A mathematical framework for the modelling of an infinite array of boson sites}
\label{Section_boson_gen_boson}

\subsection{The Fock space of an infinite array of boson sites}

 Unless stated otherwise, in the following, $p= 2,3, 5,7, \ldots$ denote primes and $n \in \mathbb{N}$ denote positive integers. Furthermore, $\delta_n$ denotes a unit measure on $\mathbb{N}$, supported at the single point $n$. Since in the construction that follows these measures acquire the significance of quantum states, they will also be denoted $\ket{n}$, i.e.,  $\ket{n} = \delta_n$.  Invoking the fundamental theorem of arithmetic, every integer $n$ is uniquely represented via
\[
n = \prod\limits_p p^{a_p(n)} \quad (\mbox{ the product is over all primes }).
\]
This defines the multiplicities $a_p(n) \in \mathbb{N}$. In what follows, we make frequent use of two arithmetic functions: $\Omega(n) = \sum_{p} a_p(n)$, and  $Q(n) = \sum_{p} a_p(n)^2$. 

We utilize a number-theoretic model of the bosonic Fock space, originated in \cite{Spector} and \cite{Bost_Connes}. It is equivalent to the standard one, yet its particular features enable us to bring to bear the Fourier analysis on $\mathbb{Q}_+$. The construction begins by defining the single-particle Hilbert space 
\begin{equation}\label{H_SP}
  \mathbb{H}_{\text{SP}} = \mbox{ span } \{\, \ket{p} = \delta_p :\, p \mbox{ prime } \},
\end{equation}
 i.e., $ \mathbb{H}_{\text{SP}}$ is here defined via a distinguished orthonormal basis with the basis vectors indexed by primes.  Next, the bosonic $k$-particle spaces are given by the symmetric  $k$-th tensor power of the single-particle space, i.e., 
\begin{equation}\label{H_kpart}
 \mathbb{H}_{\text{SP}}^{\odot 0} = \mathbb{C}, \mbox{ and }\quad \mathbb{H}_{\text{SP}}^{\odot k} = \mbox{ span}\{\, |n\rangle= \delta_n :\, \Omega (n) = k\}\quad  \mbox{ for }\quad k = 1, 2,3,\ldots.
\end{equation}
Of course, $ \mathbb{H}_{\text{SP}}^{\odot 1} =  \mathbb{H}_{\text{SP}}$.
Finally, the total Fock space is defined as the direct sum of all the $k$-particle spaces, i.e., 
\begin{equation}\label{H_Fock}
  \mathbb{H}^\odot = \bigoplus\limits_{k=0}^\infty \, \mathbb{H}_{\text{SP}}^{\odot k} .
\end{equation}
The following isomorphism of Hilbert spaces is evident:
\begin{equation}\label{H_Fock_is}
  \mathbb{H}^\odot = \mbox{ span}\{\, |n\rangle= \delta_n :\, n \in \mathbb{N}\}= \ell_2(\mathbb{N}).
\end{equation}
This is the reason for the number-theoretic construction, as identifying the Fock space with the space of square-summable sequences enables the type of analysis that is to follow.  
It is important to realize that, since the sites are distinguishable, the only universal symmetries of this structure are permutations, i.e. renumbering, of the sites. In other words, any automorphism should restrict to an isomorphism of each $N$-particle subspace. Moreover, states that are supported on any number of sites should map into states with the same property. 

\subsection{A semi-classical dynamics with Hoyleian creation of matter out of vacuum}

We will now speculate about possible dynamical properties of an array of boson sites from a semi-classical point of view. To this end consider an evolving mixed state 
\begin{equation}\label{rho_P}
  \rho_P(t) = \sum_{n} p_n(t) |n\rangle\langle n|.  
\end{equation}
It is a hermitian, positive semi-definite and trace one operator. It represents a mixed state of a boson system. It is devoid of off-diagonal terms, which makes it essentially equivalent to a classical probabilistic mixture of Fock states. Now, suppose the entries $p_n(t)$ evolve according to the process (\ref{Markov_t}). If $p_1(0)=1$, this is a stationary solution in which all mass resides in the vacuum state. The evolution is interesting if $p_1(0)<1$. The observations we have made in Section \ref{subsection_Markov} lead to the following conclusions
\begin{enumerate}
\item  If  $\rho_P(0)$ is supported on a finite sub-array, i.e. $p_n(0)$ can be nonzero only if $n$ is a product of primes form a finite set, then so is  $\rho_P(t)$. In other words, the site-support does not increase. It is already encoded in the initial condition.  

  \item The mixed state $\rho_P(t)$ evolves so that its mass is transferred from subspaces $\mathbb{H}_{\text{SP}}^{\odot k}$ with small $k$ to those with high $k$. In other words, there is no particle number preservation. One might say, the complexity of the system steadily and indefinitely increases. 
     
        \item
        In the generic case, i.e. $p_n(0)>0 $ for all $n$, the entropy is constantly increasing. 
         
  \item 
    It follows that the expectation of any observable that is supported in the direct sum of finitely many  $\mathbb{H}_{\text{SP}}^{\odot k}$ tends to zero as $t$ tends to infinity. 
    
    \item A measurement that is local, i.e., confined to finitely many sites, may have expectation that increases exponentially in time. As an example, consider the number operator (\ref{Np_arithmetic}) at site $q$. Its expectation at time $t$ is $\langle \hat{N}_q \rangle = \mbox{ trace } [\rho_P(t) \hat{N}_q ] = \sum_{k = 1}^\infty   p_{q^k}(t) k$. For the initial condition as that in the example resulting in (\ref{P_t_example}), this amounts to 
        \[
        \langle \hat{N}_q \rangle = e^{-t}\, \sum_{k=1}^{\infty} (1-e^{-t})^{k-1}  k 
        =  \sum_{k=1}^{\infty} \frac{d}{dt}  (1-e^{-t})^{k}
        = \frac{d}{dt} \frac{1-e^{-t}}{e^{-t}} = e^t.
        \]
        Also, as we have seen, the entropy is increasing linearly with time.

  \end{enumerate}
 The above properties of our postulated classical evolution are mathematically appealing, but ultimately appear to be physically unrealistic. In the remainder of this article, we consider a quantum version of this dynamics and investigate its properties, which prove to be physically plausible.

\subsection{The standard quantum dynamics of bosons}

At this stage, we define a family of annihilation and creation operators, $\hat{a}_p, \hat{a}_p^\dagger$, indexed by primes. 
\begin{equation}\label{on_arithmetic}
    \hat{a}_p \, |n\rangle = \sqrt{a_p(n)}\,\, |n/p\rangle ,\quad
       \hat{a}_p^\dagger \, |n\rangle = \sqrt{a_p(n)+1}\,\, |np\rangle.
\end{equation}
By convention, when $p$ is not a divisor of $n$, then $\ket{n/p} =0$.
A direct calculation shows that these operators satisfy the bosonic Canonical Commutation Relations (BCCR):
\begin{equation}\label{bCCR}
[\,  \hat{a}_p, \hat{a}_q \, ] = 0, \quad [\,  \hat{a}_p, \hat{a}_q^\dagger \, ] = \delta_{p,q}\quad \mbox{ for all primes } p,q.
\end{equation}
The site number operators are $\hat{N}_p = \hat{a}_p^\dagger \hat{a}_p$, and the total number operator is $\hat{N} = \sum_p \hat{N}_p$. It follows from (\ref{on_arithmetic}) that
\begin{equation}\label{Np_arithmetic}
    \hat{N}_p \, |n\rangle = a_p(n)\, |n\rangle , \quad \mbox{ and }\quad  \hat{N} \, |n\rangle = \Omega(n)\, |n\rangle .
\end{equation}
Note that $\ket{n} \in \mathbb{H}^{\odot \Omega(n)}$, i.e., $\Omega(n) =\sum_p a_p(n)$ signifies the number of particles in the state $\ket{n}$.

This framework lends itself to the modelling of an array of boson sites. Namely, consider an infinite sequence of sites, which host boson particles. The sites are indexed by consecutive primes, $2, 3, 5, 7,$ etc., rather than consecutive integers. The operator $\hat{a}_p^\dagger$ creates a particle at site $p$ and $\hat{a}_p$ annihilates a particle at site $p$. When the quantum system is in the state $\psi = \sum_{n=1}^{\infty} \, z_n \, \ket{n}$ with $\sum_{n} \, |z_n|^2 =1$, then the expected number of particles at site $p$ is $\langle \psi | \hat{N}_p \ket{\psi} =\sum_{n} a_p(n)\, |z_n|^2$. 

Note that a multi-particle system of bosons on an infinite array of sites is distinct from a boson gas. In both these systems the particles are indistinguishable from one another. However, in an array, the sites are distinguishable from one another.  
The  dynamic of such a system is typically modeled by the Bose-Hubbard Hamiltonian (BHH), or its variants. The typical BHH consists of two terms, i.e., $\mathcal{H}= \mathcal{H}_o + \mathcal{H}_h $ . The first term expresses the effect of single site energies, stemming from the Coulomb potential $U$ and the chemical potential $\mu$. It has the form 
\begin{equation}\label{NT_BH_diag}
  \mathcal{H}_o= \sum_{p} \, \frac{U}{2}\, \hat{N}_{p}(\hat{N}_{p} - 1) - \mu\, \hat{N}_{p} ,
\end{equation}
where sum is over all primes, i.e., over all sites. It acts  diagonally, i.e.
\begin{equation}\label{H0_psi}
   \mathcal{H}_o\, \sum_n z_n \ket n= \sum_n \left[\frac{U}{2} Q(n) - \left( \frac{U}{2} + \mu\right)\Omega(n)\right] \, z_n \ket n , \quad \mbox{ where } \quad Q(n) = \sum_p a_p(n)^2.
   \end{equation}
The second part is the hopping part, which assumes the form
\begin{equation}\label{NT_BH_hopp}
  \mathcal{H}_h= - t\,\sum_{\overline{pq}} (\hat{a}_{p}^\dagger \hat{a}_{q} + \hat{a}_{q}^\dagger \hat{a}_{p}).
\end{equation}
Here, the sum is over pairs of sites $\overline{pq}$. The pairing can be chosen to occur between nearest neighbours, i.e., site $2$ would be only paired to site $3$, site $3$ to both sites $2$ and $5$, site $5$ to sites $3$ and $7$, etc. However, other types of pairing can also be considered, depending on the physical characteristics of the system to be modelled. The parameter $t>0$ is the strength of hopping. The minus sign is the typical choice. The purpose of this article is to consider an alternative interaction part in place of $\mathcal{H}_h$, while retaining $\mathcal{H}_o$. 
\vspace{.5cm}

\noindent 
\emph{Remark.} The operators $\hat{a}_p$ and $\hat{a}_p^\dagger$ are represented in the standard basis by matrices with a specific ray structure. Namely, the matrix of $\hat{a}_p$ has zeros everywhere, except the ray of slope $1/p$. The matrix of $\hat{a}_p^\dagger$, is also a ray with the slope $p$,  e.g.,
\[
\hat{a}_2^\dagger =
 \left[\begin{array}{cccccccccc}
 &\cdot &\cdot &\cdot &\cdot &\cdot &\cdot &\cdot &\cdot &\ldots\\
  &1  &\cdot &\cdot &\cdot &\cdot &\cdot &\cdot &\cdot &\ldots\\
  &\cdot &\cdot &\cdot &\cdot &\cdot &\cdot &\cdot &\cdot &\ldots\\
  &\cdot & \sqrt{2} &\cdot &\cdot &\cdot &\cdot &\cdot &\cdot &\ldots\\
  &\cdot &\cdot &\cdot &\cdot &\cdot &\cdot &\cdot &\cdot &\ldots\\
  &\cdot &\cdot & 1 &\cdot &\cdot &\cdot &\cdot &\cdot &\ldots\\
  &\cdot &\cdot &\cdot &\cdot &\cdot &\cdot &\cdot &\cdot &\ldots\\
  &\cdot &\cdot &\cdot & \sqrt{3}  &\cdot &\cdot &\cdot &\cdot&\ldots\\
 &\vdots &\vdots &\vdots &\vdots &\vdots  &\vdots &\vdots &\vdots &\ddots
 \end{array}\right].
\]
Such matrices can be described by the slope of the nontrivial ray, in general, a rational number $w \in \mathbb{Q}_+$. It is easily seen that a product of two such matrices, again, has a ray structure, whose slope is the product of the slopes of the two terms. This signals the relevance of the multiplicative group of fractions and, by extension, of its dual.  

\subsection{Generalized bosons}
\label{Subsection_gen_boson}

We will also use generalized bosonic creation and annihilation operators, defined via
\begin{equation}\label{bs}
    \hat{b}_p \, |n\rangle = |n/p\rangle ,\quad
       \hat{b}_p^\dagger \, |n\rangle = |np\rangle,
\end{equation}
with the convention that $|n/p\rangle$ is replaced by $0$ in those cases when $p$ is not a divisor of $n$.  
Here, again, $p$ is any prime. The following relations are verified via direct calculations:
\begin{equation}\label{bCCR_almost}
[\,  \hat{b}_p, \hat{b}_q \, ] = 0, \quad [\,  \hat{b}_p, \hat{b}_q^\dagger \, ] = 0 \mbox{ if } p \neq q,
\end{equation}
as well as,
\begin{equation}\label{bCCR_almost2}
[\,  \hat{b}_p, \hat{b}_p^\dagger \, ] = \pi_p, \quad \mbox{ where }\,\, \pi_p \, |n\rangle = \left\{\begin{array}{ll}
                                                          0 & \mbox{ if } p \mbox{ is a divisor of } n \\
                                                           \, |n\rangle & \mbox{ otherwise } 
                                                        \end{array}\right.
\end{equation}
Note that $\pi_p$ is not an identity, but a projection onto the subspace $\mbox{ span } \{|n\rangle: p \mbox{ is not a divisor of } n\}$.  
The following lemma will play a role in what follows:
\begin{lemma}\label{Commut_X}
For a vector $x = \sum_{p} x_p |p\rangle \in \mathbb{H}_{\text{SP}}$, define operator $X = \sum_p x_p  \hat{b}_p$. Then
\[
[\,  X, X^\dagger \, ] = \sum_{p}|x_p|^2\, \pi_p,\quad \mbox{ and } \sum_{p}|x_p|^2\, \pi_p\, |n\rangle = 
\left( \sum_{p}|x_p|^2\,0^{a_p(n)}\right)\, |n\rangle.  
\]  
\end{lemma}
\begin{proof}
  The identity is seen via a direction calculation involving (\ref{bCCR_almost}) and (\ref{bCCR_almost2}). (Recall, $0^0 =1$.)
\end{proof}

It is also convenient to define operators $\hat{b}_n$ for all $n\in \mathbb{N}$. Namely, we set:
\begin{equation}\label{bns}
 \hat{b}_1 = 1,\mbox{ and }\, \hat{b}_n = \prod\limits_p (\hat{b}_p)^{a_p(n)} \, \mbox{ for all }  n >1.
\end{equation}
Since all $\hat{b}_p$ commute, $\hat{b}_n$ are well defined. The following identity is an immediate consequence of the definition
\begin{equation}\label{bkan}
 \hat{b}_m \hat{b}_n = \hat{b}_{mn} \, \mbox{ for all }  m,n.
\end{equation}
Also, it is seen directly that
\begin{equation}\label{bn_explic}
    \hat{b}_n \, |k\rangle = |k/n\rangle ,\quad
       \hat{b}_n^\dagger \, |k\rangle = |kn\rangle.
\end{equation}
Thus, the operators $\hat{b}_n$ are given in the standard basis by a matrix that has zeros everywhere, except the ray of slope $1/n$, which is filled with the constant coefficients $1$. The adjoint of an operator of this type is given by the transpose of its matrix,  e.g.,
\[
\hat{b}_2^\dagger =
 \left[\begin{array}{cccccccccc}
 &\cdot &\cdot &\cdot &\cdot &\cdot &\cdot &\cdot &\cdot &\ldots\\
  &1  &\cdot &\cdot &\cdot &\cdot &\cdot &\cdot &\cdot &\ldots\\
  &\cdot &\cdot &\cdot &\cdot &\cdot &\cdot &\cdot &\cdot &\ldots\\
  &\cdot & 1 &\cdot &\cdot &\cdot &\cdot &\cdot &\cdot &\ldots\\
  &\cdot &\cdot &\cdot &\cdot &\cdot &\cdot &\cdot &\cdot &\ldots\\
  &\cdot &\cdot & 1 &\cdot &\cdot &\cdot &\cdot &\cdot &\ldots\\
  &\cdot &\cdot &\cdot &\cdot &\cdot &\cdot &\cdot &\cdot &\ldots\\
  &\cdot &\cdot &\cdot & 1  &\cdot &\cdot &\cdot &\cdot&\ldots\\
 &\vdots &\vdots &\vdots &\vdots &\vdots  &\vdots &\vdots &\vdots &\ddots
 \end{array}\right].
\]
It is easily seen that all these operators are bounded, and their operator norms satisfy
 \begin{equation}\label{bn_norms}
   \| \hat{b}_n \| = 
      \| \hat{b}_n^\dagger \| = 1 \quad \mbox{ for all } n \in \mathbb{N}.
\end{equation}
The corresponding number operators, $\hat{M}_p = \hat{b}_p^\dagger \hat{b}_p$, and $\hat{M} = \sum_{p} \hat{M}_p$ are characterized as follows:
\begin{equation}\label{M_number}
\hat{M}_p \, |n\rangle = \left\{\begin{array}{ll}
                                                          \ket{n} & \mbox{ if } p \mbox{ is a divisor of } n \\
                                                           \, 0 & \mbox{ otherwise } 
                                                        \end{array}\right. , \quad \mbox{ and } \quad 
                                                        \hat{M} \ket{n} = \omega(n) \, \ket{n},                                                    
\end{equation} 
where $\omega(n)$ is the number of distinct prime divisors of $n$. Note that $\omega(n)$ signifies the number of sites in the support of state $\ket{n}$. It is interesting to compare (\ref{M_number}) with (\ref{Np_arithmetic}). Note also that for a state $\ket{\psi} = \sum_n x_n \ket{n}$, we have 
\begin{equation}\label{N_M_expect}
  \langle \psi | \hat{N} |\psi\rangle = \sum_{n} \Omega(n)\,|x_n|^2, \quad 
  \langle \psi | \hat{M} |\psi\rangle = \sum_{n} \omega(n)\,|x_n|^2.
\end{equation}
Thus, while $\hat{N}$ measures the expected number of particles, $\hat{M}$ measures the expected number of occupied sites.
Both sequences are characterized statistically by the Erd\H{o}s--Kac theorem; see \cite{Erdos-Kac} for the original work, \cite{Renyi-Turan} for a simplified proof with stronger error estimates, and \cite{Tenenbaum} for a textbook exposition. The theorem asserts that for $n$ chosen uniformly from
$\{1,\dots,N\}$,
\[
\frac{\omega(n)-\log\log N}{\sqrt{\log\log N}}
\Longrightarrow \mathcal N(0,1) \quad \mbox{ in the limit } N \rightarrow \infty,
\]
and the same statement holds with $\Omega(n)$ in place of $\omega(n)$.
In particular, both additive functions have asymptotic mean and variance
$\sim \log\log N$. While the proof of the theorem requires deep arguments, the fact that $\omega$ and $\Omega$ share the same limiting behavior can be seen via relatively simple arguments.

Replacing the bosonic number operators $\hat{N}_p$ with the generalized bosonic number operators $\hat{M}_p$ in (\ref{NT_BH_diag}), and observing that $\hat{M}_p^2=\hat{M}_p$, one obtains the reduced Hamiltonian for generalized bosons
\begin{equation}\label{MT_BH_diag}
  \mathcal{H}^b_o= - \mu\, \sum_p \hat{M}_{p}, \quad \mbox{ so that } \quad \mathcal{H}^b_o\, \sum_n z_n \ket n =  - \mu \sum_n \omega(n)\, z_n \ket n .
\end{equation}
Since this reduction eliminates the Coulomb interaction term entirely, one may choose instead to retain the original diagonal Bose--Hubbard term, depending on the physical regime under consideration.

\subsection{The gen-boson transform}
The model which is proposed in Section \ref{Section_New_Models} is based on the stimulated creation operator. Namely, for the state $\ket{\psi}$ we consider a spontaneously arising creation operator: 
\begin{equation}\label{Gen_boson_transform}
  \ket\psi = \sum_{n} x_n \ket n \mapsto \Bgd{\psi} = \sum_{n} x_n \hat{b}_n^\dagger \quad \mbox{ equiv. }\quad  \Bgd{\psi} \mapsto \ket \psi = \Bgd{\psi} \ket 1 .
\end{equation}
The map assigning $\Bgd{\psi}$ to $\ket\psi$ will be refered to as the gen-boson transform. Note that $\Bgd{\psi}$ can act on the state itself to induce an amplified state
\begin{equation}\label{psi_self_amp}
\Bgd{\psi}\ket\psi =  \sum_{n} x_n \, \hat{b}_n^\dagger \ket{\psi} = \sum_n\sum_m x_n x_m \ket{mn} = \sum_r \left(\sum_{d|r} x_d x_{r/d}\right)\ket{r}.
\end{equation} 
Note that the expression in the parenthesis is precisely the $r$-th term of the Dirichlet convolution of the sequence $x_n$ with itself. Furthermore,   $\Bgd{\psi}\ket\psi$ may be equivalently represented as a symmetric tensor square of $\ket{\psi}$, i.e.,
\begin{equation}\label{psi_square}
 \Bgd{\psi}\ket\psi = \ket{\psi} \odot \ket{\psi}.
\end{equation}
The gen-boson transform admits a consistent extension to the case of mixed states, see Subsection \ref{Dynamics_for_rho}. 

\subsection{Examples of action of $\Bg{\psi}\Bgd{\psi}$ and\ $\Bgd{\psi}\Bg{\psi}$ in fixed particle-number subspaces}

Let
$
|\psi\rangle = \sum_{\Omega(n)=N} x_n\,\ket n,
$
be a state with fixed particle number $N$, i.e. $\ket\psi \in \mathbb{H}_{\text{SP}}^{\odot N} $. 
Note that
\[
 \Bg{\psi}  |\psi\rangle
=
\sum_{\Omega(n)=N} x_n^*\, b_n
\sum_{\Omega(m)=N} x_m\,|m\rangle.
\]
Since $\Omega(n)=\Omega(m)=N$, the condition $n \mid m$ implies $n=m$, and hence
\[
b_n |m\rangle =
\begin{cases}
|1\rangle, & n=m,\\
0, & n\neq m.
\end{cases}
\]
Therefore
\[
 \Bg{\psi} |\psi\rangle
=
\sum_{\Omega(n)=N} |x_n|^2\, \ket 1
=
\ket 1,
\]
and consequently
\begin{equation}\label{B_triv_order}
    \Bgd{\psi} \Bg{\psi} |\psi\rangle = |\psi\rangle.
 \end{equation}
A  more interesting outcome results from reversing the order of nonlinear operations. On one hand, we have
\begin{equation}\label{BBpsi_part_numer}
  \mbox{If } \quad \ket \psi \in \mathbb{H}_{\text{SP}}^{\odot N}, \quad \mbox{ then } \quad  \Bg{\psi} \Bgd{\psi} \ket\psi\in \mathbb{H}_{\text{SP}}^{\odot N}.  
\end{equation}
(A direct argument for this identity is given in Subsection \ref{subsection_matrices}.) On the other hand, this operation is nontrivial and will introduce entanglement in some scenarios. We illustrate this with an example, as follows. Let
\[
|\psi\rangle \sim (|p\rangle + |q\rangle)\odot(|p\rangle + |q\rangle)
= |p^2\rangle + 2|pq\rangle + |q^2\rangle,
\]
so that
\[
 \Bgd{\psi} \sim b_{p^2}^\dagger + 2 b_{pq}^\dagger + b_{q^2}^\dagger = \left( b_{p}^\dagger +  b_{q}^\dagger \right)^2.
\]
A direct computation gives
\[
 \Bg{\psi} \Bgd{\psi} |\psi\rangle
\sim 
15|p^2\rangle + 20|pq\rangle + 15|q^2\rangle,
\]
i.e. operator $ \Bg{\psi} \Bgd{\psi}$ does not preserve the product structure. 
Indeed, an unentangled state has the form
$
A|p^2\rangle + B|pq\rangle + C|q^2\rangle$, 
with 
$B^2 = 4AC$.
The initial state $\ket\psi$ satisfies this condition, but the state $ \Bg{\psi} \Bgd{\psi} \ket\psi$ does not. In summary, within fixed particle-number subspaces, the nonlinear operation
\begin{equation}\label{BBpsi}
  \ket\psi \mapsto \Bg{\psi} \Bgd{\psi} \ket\psi, 
  \qquad \text{where } \ket\psi = \Bgd{\psi} \ket 1,
\end{equation}
can generically produce entanglement. However, this mechanism disappears in the completely multiplicative case. Indeed, if the state $\ket{\psi}$ factorizes across prime sectors as
\[
\ket{\psi}=\bigodot_p \ket{\psi_p},
\]
then the associated operators $\Bg{\psi}$ and $\Bgd{\psi}$ inherit the same Euler-product structure. Consequently,
\[
\Bg{\psi}\Bgd{\psi}\ket{\psi}
=
\bigodot_p
\left(
\Bg{\psi_p}\Bgd{\psi_p}\ket{\psi_p}
\right),
\]
so the evolution preserves separability between distinct prime sectors. In particular, no entanglement is generated between different primes in the completely multiplicative case.

Now consider a single prime sector
\[
\ket{\psi_p}=\sum_{k\ge0} z_k \ket{p^k}.
\]
Then
\[
\Bg{\psi_p}\Bgd{\psi_p}\ket{\psi_p}
=
\sum_{k\ge0}
\left(
\sum_{\substack{m,n,r\ge0\\ n+r-m=k}}
\overline{z_m} z_n z_r
\right)
\ket{p^k}.
\]
Under the Fourier identification
\[
\ket{p^k}\longleftrightarrow e^{2\pi i k\theta},
\qquad
\psi_p(\theta)=\sum_{k\ge0} z_k e^{2\pi i k\theta},
\]
this nonlinear operation becomes
\[
\Bg{\psi_p}\Bgd{\psi_p}\ket{\psi_p}
\quad \longleftrightarrow \quad
P_+\bigl(|\psi_p|^2\psi_p\bigr),
\]
where $P_+$ denotes projection onto the nonnegative Fourier modes. Thus, within each prime sector, the nonlinear dynamics reduces to the classical cubic Szeg\H{o}-type nonlinearity, \cite{GerardGrellierSzego}.

\section{The nonlinear models}
\label{Section_New_Models}

We introduce the Schr\"{o}dinger-Dirichlet equation (SDE), i.e., a nonlinear and nonlocal extension of the Schr\"{o}dinger equation. Its crucial term is of the form $\Bg{\psi} \Bgd{\psi} \ket\psi$ (in the notation of the previous section).

\subsection{The SDE for a pure state}\label{Subsection_SDE}

Denoting the time variable by $\tau$, we postulate SDE in the form
\begin{equation}\label{SDE_main}
  i\hslash \, \frac{\partial}{\partial\tau} \ket \psi = \mathcal{H}_o \, \ket\psi - t\,  \Bg{\psi} \Bgd{\psi} \ket\psi.
\end{equation}
As per discussion in Section \ref{Subsection_gen_boson}, it may be suitable in some modelling instances to replace $\mathcal{H}_o$ by $\mathcal{H}^b_o$. We interpret parameter $t$ as the strength of nonlocal interaction, Also the sign of $t$ may be chosen to be $\pm$ depending on the physical context. Note that potentially replacing the cubic term in (\ref{SDE_main}) with its alternative $ \Bgd{\psi} \Bg{\psi} \ket\psi$ would be of little interest because of the property (\ref{B_triv_order}). 

It is interesting to explore the Fourier-dual of (\ref{SDE_main}), via an application of harmonic analysis on the group of positive rationals, see Subsection \ref{Appendix_FT_Q_plus}. In the dual setting, the state $\ket \psi$ is represented via 
\begin{equation}\label{thePsi}
  \Psi(\vec{\mu}) = \sum_{n} x_n e^{2\pi i \vec{n}\cdot \vec{\mu}} \, \in\, H_2(\mathbb{T}^\omega, d\vec{\mu}).
\end{equation}
A direct calculation shows that 
\begin{equation}
\mathcal{H}_o\Psi  = \left(\sum_{p}\, \frac{U}{2} \left( \frac{1}{2\pi i}\frac{\partial }{\partial \mu_p}\right)^2 - \left( \frac{U}{2} + \mu\right) \frac{1}{2\pi i}\frac{\partial}{\partial \mu_p}\right)\, \Psi .
\end{equation}
Formally, $\mathcal{H}_o$ has the appearance of a differential operator. However, there is no differential structure on $\mathbb{T}^\omega$, and so the object is interpreted via (\ref{H0_psi}). Furthermore, the nonlinear operation 
(\ref{BBpsi}) is expressed in the dual form 
\begin{equation}\label{BBpsiFour}
  \Psi \mapsto  P_+ \left(|\Psi|^2\Psi \right).
\end{equation}
With this understood, the SDE equation (\ref{SDE_main}) is equivalent to 
\begin{equation}\label{SDE_main_Four}
 i \hslash \frac{\partial}{\partial\tau} \Psi =  \mathcal{H}_o \Psi - t\, P_+ \left(|\Psi|^2\Psi \right).
\end{equation}
Note that $|\Psi|^2\Psi$ need not be in $H_2(\mathbb{T}^\omega)$, even if $\Psi$ is. Thus, in order to retain the physical interpretation one needs to apply the Wiener-Hopf-Toeplitz type projection $P_+$. The nonlinear term $P_+ \left(|\Psi|^2\Psi \right)$ is well defined provided $|\Psi|^2\Psi \in L_2(\mathbb{T}^\omega)$. 
\vspace{.5cm}

\noindent
\emph{Example.} The SDE reduces to systems of ordinary differential equations of Hamiltonian type whenever a finitely supported subspace is specified as an Ansatz. To give an example, one can look for solutions in the form:
\[
\ket\psi = x(\tau)\ket{p^2} + y(\tau)\ket{pq} + z(\tau)\ket{q^2}
\]
Note that $\Omega(p^2) =\Omega(pq) = \Omega(q^2) = 2$. However, in contrast, $Q(p^2) = Q(q^2) = 4$, while $Q(pq) =2$. The SDE (\ref{SDE_main}) reduces to the system of Hamiltonian ODEs: 
\begin{equation}\label{time-three}
  \left\{\begin{array}{cccc}
    i\,\dot{x} = & (U - 2\mu)\,x & - & t\left[ (|x|^2 + 2|y|^2 + 2|z|^2)\,x  +z^* y^2\right]\\
    \\
    i\,\dot{y}  = & -2\mu \,y & - & t\left[ (2|x|^2 + |y|^2 + 2|z|^2)\,y +2 xy^*z \right]\\
    \\
    i\,\dot{z} = & (U - 2\mu)\,z & - & t\left[(2|x|^2 + 2|y|^2 + |z|^2)\,z +x^*y^2\right]
  \end{array}\right.
\end{equation}
To see the Hamiltonian property in an elementary fashion, one can verify that this system is equivalent to 
\begin{eqnarray}
 \nonumber 
  \frac{1}{2} i\, \dot{x} &=& \frac{\partial \mathcal{E}}{\partial x^*} \\
  \nonumber
 \frac{1}{2} i\, \dot{y} &=& \frac{\partial \mathcal{E}}{\partial y^*} \\
 \nonumber
  \frac{1}{2} i\, \dot{z} &=& \frac{\partial \mathcal{E}}{\partial z^*}
\end{eqnarray}
with the Hamiltonian function (energy functional):
\[
\mathcal{E} (x,y,z) = \left(\frac{U}{2} - \mu\right)\,\left(|x|^2 +|z|^2\right) - \mu\, |y|^2 - t\left[\frac{1}{4}\left( |x|^4 + |y|^4 + |z|^4 \right) + |x|^2|y|^2 + |y|^2|z|^2 + |z|^2|x|^2   + \Re \left( x^*y^2z^*\right)\right].
\]

\subsection{Extending the SDE to mixed states, no-signalling, and amplification of coherences}\label{Dynamics_for_rho}

In the previous subsection we introduced the gen-boson transform, defined via (\ref{Gen_boson_transform}), which assigns to a Fock state $\ket \psi$ a superposition of generalized-boson creation operators $X$. The nonlinear term $\Bg{\psi}\Bgd{\psi}$ then acts on $\ket \psi$. We now extend this construction to mixed states by defining the mixed-state gen-boson transform
\begin{equation}\label{Gene_boson_transform_rho}
 \rho = \sum_{m,n} \rho_{m,n}\, |m\rangle \langle n|
 \quad \mapsto \quad
 \Br{\rho}
 =
 \sum_{m,n} \rho_{m,n} \, \hat{b}_m \, \hat{b}_n^\dagger.
\end{equation}
In particular, if
$
\rho_\psi = |\psi\rangle\langle\psi|,
$
then
$
\Br{\rho_\psi}
=
\Bg{\psi}\Bgd{\psi}.
$
In this way, the evolution of a pure state given by (\ref{SDE_main_Four}) lifts naturally to the evolution of a density matrix:
\begin{equation}\label{Boson-Schrod-Mixed}
 i \hslash \dot{\rho}
 =
 \left[ \mathcal{H}_o  - t \Br{\rho}, \, \rho \right] .
\end{equation}
Indeed, it can be checked directly that if $\ket \psi$ evolves according to (\ref{SDE_main_Four}), then the projection operator $|\psi\rangle \langle \psi|$ satisfies (\ref{Boson-Schrod-Mixed}).

The nonlinear term $\left[\Br{\rho}, \, \rho \right]$ acts only locally in the following sense. Assume that the site corresponding to a prime $p$ is absent from the support of $\rho$, i.e.
\[
\rho_{m,n} \neq 0
\quad \Longrightarrow \quad
p \nmid m,n.
\]
Then $\Br{\rho}$ contains no terms involving $\hat b_p$ or $\hat b_p^\dagger$, and consequently the commutator $\left[\Br{\rho}, \, \rho \right]$ cannot generate terms involving the site $p$. Thus modifying $\rho$ locally at some sites has no effect on the dynamics at other sites. This is a \emph{no-signalling} property of (\ref{Boson-Schrod-Mixed}) and, \emph{a fortiori}, of (\ref{SDE_main_Four}) (equivalently (\ref{Boson-Schrod-discrete})).

\vspace{.2cm}

\noindent
\emph{Example.}
Consider a specific density matrix and its transform:
\[
\rho
=
\frac12\Big(
\ket{p}\bra{p}
+
\ket{p^2}\bra{p^2}
+
\ket{p}\bra{p^2}
+
\ket{p^2}\bra{p}
\Big) \quad
\mapsto \quad
\Br{\rho}
=
I + \frac12\left(\hat b_p^\dagger + \hat b_p\right),
\]
where we have applied $\hat b_n \hat b_n^\dagger = I$. 
A direct computation gives
\[
\left[\Br{\rho}, \, \rho \right]
=
\frac14\Big(
\ket{p^3}\bra{p}
-
\ket{p}\bra{p^3}
+
\ket{1}\bra{p^2}
-
\ket{p^2}\bra{1}
+
\ket{1}\bra{p}
-
\ket{p}\bra{1}
+
\ket{p^3}\bra{p^2}
-
\ket{p^2}\bra{p^3}
\Big).
\]
This example illustrates how the evolution (\ref{Boson-Schrod-Mixed}) generates new local coherences from pre-existing ones, while remaining confined to the same prime sector.

\subsection{SDE as a matrix equation}\label{subsection_matrices}

In this subsection we express SDE as an infinite matrix equation, which helps in finding examples of solutions. 
To fix notation, let $\Psi$ be as in (\ref{thePsi}), and let in addition
\begin{equation}\label{thePhi}
  \Phi(\vec{\mu}) = \sum_{n\in \mathbb{N}} y_n e^{2\pi i \vec{n}\cdot \vec{\mu}}
\end{equation}
Furthermore, to the vector column $x = [x_1, x_2, x_3, \ldots ]^T$, we assign the matrix $X$ of $\Bgd{x} = \sum_{n} x_n \, \hat{b}_n^\dagger$ in the standard basis:
\begin{equation}\label{theX}
X =
 \left[\begin{array}{cccccccccc}
 &x_1 &\cdot &\cdot &\cdot &\cdot &\cdot &\cdot &\cdot &\ldots\\
  &x_2  &x_1 &\cdot &\cdot &\cdot &\cdot &\cdot &\cdot &\ldots\\
  &x_3 &\cdot &x_1 &\cdot &\cdot &\cdot &\cdot &\cdot &\ldots\\
  &x_4 &x_2 &\cdot &x_1 &\cdot &\cdot &\cdot &\cdot &\ldots\\
  &x_5 &\cdot &\cdot &\cdot &x_1 &\cdot &\cdot &\cdot &\ldots\\
  &x_6 &x_3 &x_2 &\cdot &\cdot &x_1 &\cdot &\cdot &\ldots\\
  &x_7 &\cdot &\cdot &\cdot &\cdot &\cdot &x_1 &\cdot &\ldots\\
  &x_8 &x_4 &\cdot &x_2  &\cdot &\cdot &\cdot &x_1&\ldots\\
 &\vdots &\vdots &\vdots &\vdots &\vdots  &\vdots &\vdots &\vdots &\ddots
 \end{array}\right] = \sum_{n=1}^{\infty} x_n \, \hat{b}_n^\dagger,
\end{equation}
where $\hat{b}_n^\dagger$ are as in (\ref{bns}).  Similarly, to the vector column $y = [y_1, y_2, y_3, \ldots ]^T$, we assign the infinite matrix $Y = \sum_{n} y_n \, \hat{b}_n^\dagger$, etc. The following two identities are obtained by direct calculation:
\begin{equation}\label{phipsi}
  \Phi\Psi = \sum_{n\in \mathbb{N}} z_n e^{2\pi i \vec{n}\cdot \vec{\mu}},\quad \mbox{ where } z = Y x,
\end{equation}
\begin{equation}\label{phibarpsi}
  P_+\left(\Phi^*\Psi\right) = \sum_{n\in \mathbb{N}} w_n e^{2\pi i \vec{n}\cdot \vec{\mu}},\quad \mbox{ where } w = Y^\dagger x.
\end{equation}
\begin{lemma} \label{Lem_part_num}
Retaining the notation as above, assume $\Phi \in \mathbb{H}^{\odot k}$ and $\Psi \in \mathbb{H}^{\odot l}$. Then 
\begin{enumerate}
  \item $\Phi\Psi  \in \mathbb{H}^{\odot k+l}$.
  \item $P_+\left(\Phi^*\Psi\right)  \in \mathbb{H}^{\odot l-k}$.
\end{enumerate}
\end{lemma}

\begin{proof}
\emph{1.} Note $z_n = \sum_{d|n} y_d x_{n/d}$. At the same time, $y_d x_{n/d} \neq 0 $ implies $\Omega(d) = k$ and $\Omega(n/d) = l$, i.e. $\Omega(n) = k+l$. 

\emph{2.} Note $w_n = \sum_{r} y_r^* x_{rn}$. At the same time, $y_r^* x_{rn} \neq 0 $ implies $\Omega(r) = k$ and $\Omega(rn) = l$, i.e. $\Omega(n) = l-k$.
  \end{proof}
Identities (\ref{phipsi}) and (\ref{phibarpsi}) imply:
\begin{equation}\label{psibarpsipsi}
  P_+\left(|\Psi|^2\Psi\right) = \sum_{n\in \mathbb{N}} u_n e^{2\pi i \vec{n}\cdot \vec{\mu}},\quad \mbox{ where } u = X^\dagger X x.
\end{equation}
 In particular, equation  (\ref{SDE_main}), equivalently (\ref{SDE_main_Four}), can be expressed in the form:
\begin{equation}\label{Boson-Schrod-discrete}
i \hslash \dot{x} =  \mathcal{H}_o  x - t \,  X^\dagger X  x.
\end{equation}
From Lemma \ref{Lem_part_num}, we obtain:
\begin{corollary}
If  $\Psi \in \mathbb{H}^{\odot k}$ (equiv. $\ket\psi \in \mathbb{H}^{\odot k}$), then $ P_+ \left(|\Psi|^2\Psi \right)  \in \mathbb{H}^{\odot k}$ (equiv. $ \Bg{\psi} \Bgd{\psi} \ket\psi \in \mathbb{H}^{\odot k} $).
\end{corollary}

Note that (\ref{Boson-Schrod-discrete}) is the Euler-Lagrange equation for the energy functional: 
\begin{equation}\label{Energy}
  \mathcal{E} = \langle x\,|\,\mathcal{H}_o x \rangle - \frac{t}{4}\, \| X x \|^2.
\end{equation}
In the notation of Subsection \ref{Subsection_gen_boson}, we have $\ket{\psi} = \sum_n x_n \ket{n}$. The nonlinear term, which represents spontaneous amplification, can be equivalently represented in the forms:
\begin{equation}\label{Nonlin_equiv}
   \| X x \|^2 = \left\| \Bgd{\psi}\, \ket{\psi} \right\|^2 
  = \left\| \, \psi\odot \psi \right\|^2.
\end{equation}
For brevity we use a short-hand notation $\left\| \, \ket{\psi}\odot \ket{\phi} \right\| = \left\| \psi\odot \phi \right\|$.

\subsection{Examples of stationary solutions for a fixed number of particles}\label{Subsect_finite_sols}

In what follows we display specific eigenmodes, using special assumptions. We will also display an example which evidences that these are not all possible solutions. Some motivation for the underlying assumptions is given in  Section \ref{Section_Appendix}. 

The main problem is to analyze the expression $\Bg{\psi}\Bgd{\psi}\ket\psi $.
We consider the Ansatz:
\begin{equation}\label{the_x_assumpt}
 \ket{ \psi } = \sum_{j = 1}^{m} x_{n_j} \ket{n_j}, \, \mbox{ where: } \Omega(n_j) = N, \mbox{ and all pairs of } n_j
  \mbox{ are mutually prime. } .
\end{equation}  
We also assume normalization $\sum_{j} \, |x_{n_j}|^2 =1$. Note that the distinct states $\ket{n_j}$ do not overlap at any sites. 
Next,
\[
\Bgd{\psi}\ket\psi =  \sum_{k=1}^m x_{n_k} \, \hat{b}_{n_k}^\dagger\, \sum_{j = 1}^{m} x_{n_j} \ket{n_j} = \sum_{j, k=1}^m x_{n_j} x_{n_k}\, \ket{ n_j n_k}.
\]
Subsequently, 
\[
\begin{array}{ccc}
\Bg{\psi}\Bgd{\psi}\ket\psi & = &  \sum_{l=1}^m x_{n_l}^* \, \hat{b}_{n_l}\, \sum_{j, k=1}^m x_{n_j} x_{n_k}\, \ket{ n_j n_k} \\
  && \\
   & =  &  \sum_{j, k, l =1}^m x_{n_l}^*x_{n_j} x_{n_k}\, \ket{ n_j n_k/n_l}\\
   && \\
    & = &  \sum_{k=1}^m \left( 2\|x\|^2 -  |x_{n_k}|^2 \right)\, x_{n_k}\, \ket{ n_k}\\
\end{array}
\]
 where we have used the fact that $n_j n_k/n_l$ is nonzero only when $l = j$ or $l=k$. In summary,
 \begin{equation}\label{XdXx}
\Bg{\psi}\Bgd{\psi}\ket\psi  =    \sum_{k=1}^m \left( 2 - |x_{n_k}|^2 \right)\, x_{n_k}\,\ket{ n_k} \mbox{ under the assumptions (\ref{the_x_assumpt})}.
 \end{equation} 
  It follows that the stationary solutions of (\ref{SDE_main}) are characterized by the system of equations 
\begin{equation}\label{E_eq}
  \frac{U}{2} Q(n_k) - \left(\frac{U}{2} + \mu \right) N  - t \, \left( 2-  |x_{n_k}|^2 \right) = E, \quad  k = 1,2,\ldots m.
\end{equation}
Summing over $k$ and solving for $E$, we obtain
\begin{equation}\label{the_E}
  E = \frac{U}{2} \langle Q\rangle - \left(\frac{U}{2} + \mu\right)N - t\left(2-\frac{1}{m}\right), \quad \mbox{ where } \, 
  \langle Q\rangle = \frac{1}{m} \sum_{k=1}^{m} \, Q(n_k).
\end{equation}
Moreover, substituting this into (\ref{E_eq}), we obtain close-form expressions for the probability distribution:
\begin{equation}\label{amplit}
  |x_{n_k}|^2 = \frac{1}{m} + \frac{U}{2t}\left[\langle Q \rangle - Q(n_k)\right], \quad  k = 1,2,\ldots m.
\end{equation}
Suppose $t>0$, and consider a measurement of the quasi-particle state $\ket{\psi}$ in the Fock basis. It is more likely to find it in a state $\ket{n_k}$ with lower $Q(n_k)$. For example, when $N = 5$, the state $\ket{n} = \ket{ 2^1 3^1 5^1 7^1 11^1}$ corresponds to $Q(n) = 5$, whereas $\ket{n} = \ket{7^5}$ corresponds to $Q(n) = 25$. In this case, $\langle Q \rangle \simeq 12.43$, and so the maximal difference $\langle Q \rangle - Q(7^5) \simeq -12.57$. This imposes a constraint on $U/(2t)$, depending on $m$, via the requirement that the probabilities (\ref{amplit}) must be non-negative.
\vspace{.5cm}

\noindent
\emph{Remark 1.}
General solutions in the single-particle subspace $\mathbb{H}_{\text{SP}}$ satisfy Ansatz (\ref{the_x_assumpt}). Indeed, in this case $\ket{ \psi} = \sum_{p \in P} x_{p} \ket{p}$, where $P$ is a finite, $m$-element, subset of the set of primes. Of course,  $\Omega(p) = Q(p) = 1$. In this case, (\ref{the_E}) assumes a simpler form:
\begin{equation}\label{the_Ep}
  E = - \mu - t\left(2-\frac{1}{m}\right), \quad \mbox{ where }  
 \end{equation} 
while (\ref{amplit}) is reduced to the statement $|x_{p}|^2 = \frac{1}{m}$, where $m$ is the number of nonzero amplitudes (the length of the support of the eigenstate). It is easily seen that these are the only solutions in the single-particle space. 
 \vspace{.5cm}

\noindent
\emph{Remark 2.} Replacing $\mathcal{H}_o$ by $\mathcal{H}^b_o$ in the SDE as well as modifying condition (\ref{the_x_assumpt}) by requiring $\omega(n_j) = N$ rather than $\Omega(n_j) = N$, one obtains the constant solution $|x_{n_j}|^2= 1/m$, similarly to the case of the previous remark.

\section{The energy functional}\label{Section_minimizer}

In this section we consider the question of existence of the global minimizer of the energy functional. The setting is that of the microcanonical ensemble, wherein the minimizer is sought in the fixed particle-number subspace of the Fock space.

\subsection{Analytical foundations}

For a general Fock state, i.e. $\ket\psi = \sum_{n=1}^{\infty} x_n \ket{n}$, we have
\begin{equation}\label{dotodot_norm_standard}
  \|\psi \odot \psi\|^2 = \sum_{n} \left|\sum_{d|n} x_d x_{n/d}\right|^2.
  \end{equation}
Note that the generalized Fourier transform assigns $\ket{\psi} \rightarrow \Psi (\vec{\mu}) $ given by (\ref{thePsi}).  Moreover, an inspection shows that
\begin{equation}\label{psidotpsiFour}
  \ket{\psi} \odot \ket{\psi} \rightarrow \Psi (\vec{\mu})^2.
\end{equation}
In particular,  by Cauchy-Schwartz inequality, we have
\[ 
\|\psi\|^2 = \int_{\mathbb{T}^\omega} |\Psi(\vec{\mu})|^2 \, d\vec{\mu} \leq \sqrt{\int_{\mathbb{T}^\omega} |\Psi(\vec{\mu})|^4 \, d\vec{\mu}} \sqrt{\int_{\mathbb{T}^\omega} 1\, d\vec{\mu} } = 
\sqrt{\int_{\mathbb{T}^\omega} |\Psi(\vec{\mu})^2|^2 \, d\vec{\mu}} = \|\psi \odot \psi\|.
\]
In summary,
\begin{equation}\label{psidotpsinormsup}
  \|\psi\|^2 \leq\|\psi \odot \psi\|,
\end{equation}  
and the two sides are equal if and only if $\Psi (\vec{\mu}) = \exp(2 \pi i \vec{n}\cdot\vec{\mu})$ (a monomial). Therefore the problem of finding the minimum of the functional $\ket{\psi} \mapsto\|\psi \odot \psi\|^2$ is trivial: the minimum is attained on monomials. 

On the other hand, for a general Fock state,  $\|\psi \odot \psi\|^2$ is not bounded above on the unit sphere $\{\ket{\psi}: \|\psi\|=1\}$. Indeed, to this end let $\ket\psi = n^{-1/2} \sum_{k=1}^{n} \ket{k}$, and let $d(m)$ denote the number of divisors of integer $m$. Thus, (\ref{dotodot_norm_standard}), and convexity of the parabola, give
\[
 \|\psi \odot \psi\|^2 \, \geq \, \frac{1}{n}\sum_{k\leq n} d(k)^2 \geq 
\left(\frac{1}{n}\sum_{k\leq n} d(k)\right)^2 
 \sim \log^2 n, 
\]  
which diverges to infinity with $n$. Here, we have used a well-known asymptotic for the growth of the summatory function of $d(m)$, e.g. see \cite{Apostol2}.

However, the functional is bounded on subspaces $\mathbb{H}_{\text{SP}}^{\odot N}$, where the number of particles is fixed at $N$. Indeed, we have the following
\begin{lemma}\label{lemma_quadest}
Let $\ket\psi, \ket\phi \in \mathbb{H}_{\text{SP}}^{\odot N}$ (not necessarily normalized). Then,
\[
  \|\psi \odot \phi\|^2  \leq \binom{2N}{N} \, \|\psi\|^2 \|\phi\|^2.
\]
\end{lemma}
\begin{proof}
  Let
$
\ket\psi = \sum_{n:\Omega(n) = N} x_n \ket{n}, \ket\phi = \sum_{n:\Omega(n) = N} y_n \ket{n}
$
Note that $\ket\psi \odot \ket\phi \in \mathbb{H}_{\text{SP}}^{\odot 2N}$. The number of possible representations of $m$ with $\Omega (m) = 2N$ as a product, say, $m=ab$, where $\Omega(a) = \Omega(b) =N$  is at most $\binom{2N}{N}$. (It is exactly  $\binom{2N}{N}$ precisely when $m$ is a product of 2N different primes.) Also, note that 
\begin{equation}\label{sum_dual}
  \sum_m \sum_{d|m} |x_d|^2 |y_{m/d}|^2
 =  \sum_{n} |x_n|^2 \,  \sum_{n} |y_n|^2.
\end{equation}
Indeed, the left hand side is equal to the action of $\sum_{n} |x_n|^2 \hat{b}_n^\dagger$ on the vector $[ |y_1|^2, |y_2|^2, |y_3|^2, \ldots ]^T$, followed by summation of all terms of the resulting vector. This is equivalent to the sum of terms of the form $|x_n|^2\,\sum_{m=1}^\infty |y_m|^2$. 
Next, (\ref{sum_dual}) and the convexity of the parabola imply 
\begin{equation}\label{odot_norm_est}
   \|\psi \odot \phi\|^2 \, = \sum_{m} \left|\sum_{d|m} x_d y_{m/d}\right|^2 
 \leq \sum_{m} \binom{2N}{N} \sum_{d|m} |x_d|^2 |y_{m/d}|^2
 = \binom{2N}{N} \sum_{n} |x_n|^2 \, \sum_{n} |y_n|^2,
\end{equation}
where we have used the convention that $x_n=0$ whenever $\Omega(n) \neq N$.  This proves the lemma.
\end{proof}

With this understood, consider the functional
\begin{equation}\label{naive_BH}
 \mathbb{H}_{\text{SP}}^{\odot N} \ni \ket\psi \mapsto  \mathcal{E}[\psi] = \langle \psi |\mathcal{H}_0 |\psi \rangle - \frac{t}{4}\|\psi \odot \psi\|^2, \quad \mbox{ where } t> 0,
\end{equation} 
where $\mathcal{H}_0 $ is as in (\ref{H0_psi}) with the potentials $U\geq 0 $ and $\mu\geq 0$. Since $Q(n) \leq \Omega(n)^2 = N^2$, the first term, $\langle \psi |\mathcal{H}_0 |\psi \rangle$, is bounded on the unit sphere in $\mathbb{H}_{\text{SP}}^{\odot N}$. In light of Lemma \ref{lemma_quadest} the disorder term is also bounded. Summarizing,
\begin{equation}\label{E_below_first}
 -N (U/2+\mu) -  (t/4)\binom{2N}{N}\, \leq \,   \mathcal{E}[\psi] \, \leq \, N^2U/2,\, \quad \mbox{ whenever } \|\psi\| =1.
\end{equation}
If, in addition, the underlying array of boson sites were finite (i.e., only a finite number of primes were at play), the functional would be reduced to a smooth function on a finite-dimensional sphere. Obviously, in such a case there is at least one state where the function attains its minimal value.
 
 In the case of finite arrays,  the nature of states that minimize this energy functional depends intricately on $N$, as well as on the number of sites; this is illustrated in Subsection \ref{subsection_three_on2} where it is found, via calculus, that when $\mu = U = 0$, the two-particle minimizer on two sites is (up to the phase factors) $\ket\psi = 7^{-1/2} \left( \sqrt{2} \ket{p_1^2} + \sqrt{3} \ket{p_1p_2} + \sqrt{2} \ket{p_2^2} \right)$. In addition, simple finite-array solutions can be obtained by utilizing Ansatz (\ref{the_x_assumpt}), see Subsection \ref{Appendix_coprime}. 
 
However, we are mainly interested in a general setting that does not introduce any a priori restrictions. The following example illustrates the nature of the problem of finding minima of $\mathcal{E}[\psi]$:
\vspace{.5cm}

\noindent \emph{Example.}
 Consider $\mathcal{E}[\psi]$ in the single-particle space, i.e. $\ket\psi \in \mathbb{H}_{SP}$. It is interesting to examine the special states given at the end of Subsection \ref{Subsect_finite_sols}, say, $\ket{ x_m } = m^{-1/2}\,\sum_{j=1}^m  \ket{p_j}$, where $\{p_j; j = 1,2,\ldots m\}$ is an arbitrary m-tuple of distinct primes. A direct calculation shows that $\|\psi_m \odot \psi_m\| = 2-1/m $. In light of inequality (\ref{odot_norm_est}) these norms approach the upper bound of $\binom{2}{1} = 2$. Furthermore,  
\[
\mathcal{E}[\psi_m] = - \mu - \frac{t}{2} + \frac{t}{4m} \rightarrow -\mu - \frac{t}{2}\quad \mbox{ as } m \rightarrow \infty.
\]
Thus, the sequence $\psi_m$ is the minimizing sequence of the functional. At the same time, there seems to be little hope of finding a minimizer state. 
\vspace{.5cm}

  To ensure the existence of minimizing states in the infinite support setting we will modify the functional by replacing the standard norms by a type of tempered norms. We present this approach in the following sections.

\subsection{The tempered norms}
The problem of finding maxima becomes tractable when the standard norm used in the definition of $\mathcal{E}[\psi]$ is replaced with a suitably tempered norm.  
To set the ground, we invoke weighted Hilbert spaces akin to the Sobolev spaces, \cite{Sowa2017}. Namely, for $\sigma \in \mathbb R$ define
\[
\mathrm{h}^\sigma = \{f:\mathbb N \rightarrow \mathbb C: \|f\|_\sigma: = \left(\,\sum\limits_{n\in\mathbb{N}} |f(n)|^2\, n^{2\sigma}\,\right)^{1/2} <\infty  \},
\]
with the inner product  $\langle f | g\rangle_\sigma = \sum_{n} f(n)^*\, g(n)\, n^{2\sigma}$ .
The following facts are easily established by standard arguments:
\begin{itemize}

\item
The collection of vectors $e_n = n^{-\sigma}\delta_n$ furnish an orthonormal basis in $\mathrm{h}^\sigma$. Furthermore, this basis furnishes an isomorphism $\mathrm{h}^\sigma \longleftrightarrow \ell_2$. In particular, $\mathrm{h}^\sigma $ is a (complete) Hilbert space.

\item
The dual space of $\mathrm{h}^\sigma$ is $(\mathrm{h}^\sigma)^* = \mathrm{h}^{-\sigma}$.

\item
The following analogue of the Rellich Compactness Theorem holds: If $\sigma_1>\sigma_2$, then the natural inclusion $I: \mathrm{h}^{\sigma_1} \hookrightarrow \mathrm{h}^{\sigma_2}$ is a compact operator. Indeed, let $I_N$ be the finite rank operator
given by $I_N[f](n) = f(n)$ for $n\leq N$, and $I_N[f](n) = 0$ for $n> N$. It is easily seen that $\|I - I_N\|_{\mathrm{h}^{\sigma_1}\rightarrow \mathrm{h}^{\sigma_2}} \rightarrow 0$ as $N \rightarrow \infty$.
\end{itemize}

We now turn attention to the functional 
\begin{equation}\label{temperedpsidotpsi}
  \ket{\psi} \mapsto\|\psi \odot \psi\|_{-\sigma}^2, \quad \mbox{ where } \sigma > 1/2.
\end{equation}
Let $\ket{\psi}, \ket{\phi} \in \mathbb{H}_{SP}^{\odot N}$. Observe,
\[
\|\psi \odot \phi\|_{-\sigma}^2 
= \sum_{m }  \left|\, \sum_{d|m} x_d y_{m/d}\right|^2 m^{-2\sigma}
= \sum_{m }  \left|\, \sum_{d|m} x_dd^{-\sigma} y_{m/d}\left(\frac{m}{d}\right)^{-\sigma}\right|^2
\]
Thus, using Lemma \ref{lemma_quadest} with substitutions $x_n \mapsto x_n n^{-\sigma}$, $y_n \mapsto y_n n^{-\sigma}$, we obtain
\begin{equation}\label{odot_sigma_norm_estim}
  \|\psi \odot \phi\|_{-\sigma}^2  \leq \binom{2N}{N} \, \|\psi\|_{-\sigma}^2 \|\phi\|_{-\sigma}^2.
\end{equation}
This allows us to demonstrate the following:
\begin{lemma} \label{weak_conv_prod}
 Let $\ket{\psi_k}\in \mathbb{H}_{SP}^{\odot N}$, $k = 1,2,\ldots $ be a sequence of states, which converges weakly to the limit $\ket{\psi_\infty}$. Then, the sequence $\psi_k\odot\psi_k$ converges strongly to $\psi_\infty\odot \psi_\infty$ in $h^{-\sigma}$ for all $\sigma >0$. 
\end{lemma}
\begin{proof}
Fix arbitrary $\sigma > 0 $.
 Since the inclusion $h^0 \hookrightarrow h^{-\sigma}$ is compact, we have 
  \[
  \|\psi_k -\psi_\infty\|_{-\sigma} \rightarrow 0 \quad \mbox{ as } k \rightarrow \infty .
  \]
    Using (\ref{odot_sigma_norm_estim}), we obtain
\begin{eqnarray*}
  \|\psi_\infty\odot\psi_\infty - \psi_k\odot\psi_k\|_{-\sigma}  &=&  \|\psi_\infty\odot(\psi_\infty-\psi_k) - \psi_k\odot(\psi_k -\psi_\infty)\|_{-\sigma} \\
  && \\
   &\leq & \|\psi_\infty\odot(\psi_k-\psi_\infty)\|_{-\sigma} + \|\psi_k\odot(\psi_k -\psi_\infty)\|_{-\sigma} \\
   &&\\
  &\leq & \binom{2N}{N}\, \left(\, \|\psi_\infty\|_{-\sigma} + 
   \|\psi_k\|_{-\sigma} \, \right)\, \|\psi_k -\psi_\infty\|_{-\sigma}.
\end{eqnarray*}
Note that $\|\psi_k\|_{-\sigma
}  \leq \|\psi_k\| = 1$. Furthermore, $\|\psi_\infty\|_{-\sigma} \leq \|\psi_\infty\| \leq \liminf \|\psi_k\| = 1$, since norms are lower semicontinuous under weak convergence. 
 Therefore, $  \|\psi_\infty\odot\psi_\infty - \psi_k\odot\psi_k\|_{-\sigma}  \rightarrow 0$. This completes the proof.  
\end{proof}

\subsection{The tempered energy functional} 

Since the functional $\mathcal{E}[\psi]$ is bounded below, there exists a sequence of states $\ket{\psi_k}$ such that $\mathcal{E}[\psi_k]$ tends to the minimum. However, that does not guarantee that there exist states where the energy is exactly the minimum. Indeed, while $\ket{\psi_k}$ has a weak limit, it does not necessarily have a strong limit. In order to prove the existence of a limit that is the minimizer, it appears necessary to have some sort of compactness. 
Therefore, it is interesting to consider a modification of  $\mathcal{E}[\psi]$  which will ensure that and to ensure the existence of minimizer states. To this end 
 consider the energy functional:
\begin{equation}\label{tempered_BH}
     \mathbb{H}_{\text{SP}}^{\odot N} \ni \ket\psi \mapsto  \mathcal{E}_\sigma[\psi] = \langle \psi |\mathcal{H}_0 |\psi \rangle_{-\sigma} - \frac{t}{4}\|\psi \odot \psi\|_{-\sigma}^2, \quad \mbox{ where } t> 0, \mbox{ and }  \sigma> 0.  
\end{equation} 
We are interested in the problem of existence of the minimum of $\mathcal{E}_\sigma[\psi]$ on the unit sphere determined by $\|\psi\|= \|\psi\|_0 = 1$ (i.e., in the standard norm) with the additional condition $\ket{\psi} \in \mathbb{H}^{\odot N} $. In other words we require that $\ket{\psi}$ be an $N$-particle state. Naturally, the use of the standard norm in the condition  $\|\psi\|=1$  is dictated by the structure of quantum mechanics, it could not be revised. On the other hand, the use of tempered norms to define the energy is not un-physical. The weights introduce a type of screening; in other words, the energy contribution from distant parts, i.e., when $n$ is large are discounted.  The Bose-Hubbard part consists of a negative and a positive part (assuming $U, \mu \geq 0$):
\begin{equation}\label{theH0part}
\langle \psi |\mathcal{H}_0 |\psi \rangle_{-\sigma} = \sum_{n: \Omega(n)= N} \left[ \frac{U}{2}Q(n) - \left( \frac{U}{2}+\mu\right)\Omega(n) \right]\,|x_n|^2 n^{-2\sigma}
= -N \left( \frac{U}{2}+\mu\right) \, \|\psi\|_{-\sigma}^2 + \frac{U}{2}\sum_{n: \Omega(n)= N} Q(n) \,|x_n|^2 n^{-2\sigma}.
\end{equation}

The negative part contributes more, when $\ket\psi$ is supported on sites with a small label, i.e. when it is local. However, the positive part introduces a ``penalty" for those $\ket\psi$ which accumulate more mass on a single site, e.g., $\ket{2^N}$ gives the highest possible positive contribution: $N^2U\,2^{-2\sigma-1}$. 
Using (\ref{odot_norm_est}), we have 
\begin{equation}\label{odot_sigma_est}
 \|\psi \odot \psi\|_{-\sigma}^2 \leq\|\psi \odot \psi\|^2 \leq \binom{2N}{N}.
\end{equation}
The contribution from $\|\ket{\psi} \odot \ket{\psi}\|_{-\sigma}^2$ is also greater when $\ket\psi$ is local.  
%
%

\subsection{Existence of an energy-minimizing state}
In this subsection, we prove that there exists a state, which minimizes the functional (\ref{tempered_BH}). Namely, we have

\begin{theorem}\label{theorem_minimizer}
    Let $U,\,\mu >0$ be fixed. For arbitrary $\sigma >0$, and $N \in \mathbb{N}$, the energy functional (\ref{tempered_BH}) attains its minimum value, i.e. 
    \[ \mbox{ there exists } \, \ket{\psi_\infty} \in \mathbb{H}_{\text{SP}}^{\odot N}\, \mbox{ with } \|\psi_\infty\| =1, \,\, \mbox{ such that } \, 
    \mathcal{E}_\sigma[\psi_\infty]= \inf\,\,\left\{\mathcal{E}_\sigma[\psi]\, : \,  \|\psi\|=1, \ket{\psi} \in \mathbb{H}_{\text{SP}}^{\odot N} \right\}.
\]
\end{theorem}

\begin{proof}

Note that $\|\psi\|_{-\sigma} \leq \|\psi\| = 1$, so the first term on the right-hand side of (\ref{theH0part}) is finite and bounded below by $-N (U/2+\mu) $. Additionally, the second term on the right of (\ref{theH0part})  is positive and,  since $Q(n) \leq N^2$, it is bounded above by $N^2U/2$. Using (\ref{odot_sigma_est}) we can summarize as follows:
\begin{equation}\label{E_below}
 -N (U/2+\mu) -  (t/4)\binom{2N}{N}\, \leq \,   \mathcal{E}_\sigma[\psi] \, \leq \, N^2U/2,\, \quad \mbox{ whenever } \|\psi\| =1.
\end{equation}
This mirrors estimates (\ref{E_below_first}).
Moreover, 
\[ \mathcal{E}_{\sigma, 0}:= \inf_{\ket{\psi}: \|\psi\|=1}\mathcal{E}_\sigma[\psi]
\]
 is a strictly negative number. Indeed, it suffices to take $\ket{\psi} = \sum x_n \ket{n}$, such that $x_n \neq 0$ implies $Q(n) = \Omega(n) =N$, i.e. nontrivial $x_n$ are only found at Fock basis states $\ket{n}$ with no more than one particle per site. In such a case $\langle \psi |\mathcal{H}_0 |\psi \rangle_{-\sigma} = -\mu N\, \|\psi\|_{-\sigma}$. Thus, $\mathcal{E}_\sigma[\psi] <0$, \emph{a fortiori} $\mathcal{E}_{\sigma, 0} <0$. 

Next, we will demonstrate that there is a minimizer $\ket{\psi_{\infty}}$ of $\mathcal{E}_\sigma[\psi] $ such that $\|\psi_{\infty}\|_{-\sigma} \leq \|\psi_{\infty}\| = 1$.   Indeed, pick a sequence $\ket{\psi_{k}}$, $k = 1,2, \ldots$, such that $\|\psi_k\|=1$, and 
\[\mathcal{E}_\sigma[\psi_k] \rightarrow \mathcal{E}_{\sigma, 0} <0. 
\]
The unit sphere in $h^0$, i.e. $\{\ket{\psi}:\|\psi\|=1\}$, is weakly compact. Therefore there exists a subsequence (also denoted $\ket{\psi_{k}}$), which weakly converges to a certain $\ket{\psi_{\infty}}\in h^0$. 
Since the immersion  $h^0 \hookrightarrow h^{-\sigma}$ is compact, the sequence $\ket{\psi_k}$ converges strongly to $\ket{\psi_{\infty}}$ in $h^{-\sigma}$, i.e.    
\[
\| \psi_k -\psi_\infty \|_{-\sigma} \rightarrow 0, \quad \mbox{ and, automatically, }\quad   \|\psi_k\|_{-\sigma } \rightarrow \|\psi_\infty\|_{-\sigma }.
\]  
Lemma \ref{weak_conv_prod} implies 
\begin{equation}\label{quad_conv}
  \|\psi_k \odot \psi_k\|_{-\sigma } \rightarrow \|\psi_\infty \odot \psi_\infty\|_{-\sigma }.
\end{equation}
Furthermore, Let $\ket{\psi_k} = \sum_n x_n^{(k)} \ket{n} $, and $\ket{\psi_\infty} = \sum_n x_n^{(\infty)} \ket{n} $. Since  $\sum_{n} \left|x_n^{(k)} - x_n^{(\infty)}\right|^2 Q(n)\, n^{-2\sigma} \leq N^2 \, \|\psi_k - \psi_\infty\|_{-\sigma} \rightarrow 0$, one has 
\[
\sum_{n} \left|x_n^{(k)}\right|^2 Q(n)\, n^{-2\sigma} \rightarrow \sum_{n} \left|x_n^{(\infty)}\right|^2 Q(n) n^{-2\sigma}.
\]
Thus, (\ref{theH0part}) implies
\begin{equation}\label{lin_conv}
  \langle \psi_k |\mathcal{H}_0 |\psi_k \rangle_{-\sigma} \rightarrow
\langle \psi_\infty |\mathcal{H}_0 |\psi_\infty \rangle_{-\sigma}.
\end{equation}
Collecting (\ref{quad_conv}) and (\ref{lin_conv}) we see that 
 $ \mathcal{E}_\sigma[\psi_k ] \rightarrow  \mathcal{E}_\sigma[\psi_\infty ].$
Therefore,
\begin{equation}\label{final}
  \mathcal{E}_\sigma[\psi_\infty ] = \mathcal{E}_{\sigma, 0}.
\end{equation}
Finally, we need to demonstrate that  $\|\psi_\infty\|=1$. To this end, note that by weak convergence 
we have an \emph{a priori} estimate
\begin{equation}\label{limitnorm}
  \| \psi_\infty \| \leq \liminf \|\psi_k\| = 1.
\end{equation} 
Since $\mathcal{E}_{\sigma, 0} <0$, we also know that $\ket\psi_\infty \neq 0$, i.e. $\|\psi_\infty\|>0$. 
Furthermore, for $c\in \mathbb{R}$, define 
\[
f(c)  := \langle \psi_\infty |\mathcal{H}_0 |\psi_\infty \rangle_{-\sigma}\,c - \frac{t}{4} \|\ket{\psi_\infty} \odot \ket{\psi_\infty}\|_{-\sigma}^2\, c^2.
\]
Note that when $c>0$, $f(c) = \mathcal{E}_\sigma[\sqrt{c}\psi_\infty]$.
The graph of $c \mapsto f(c)$ is an inverted parabola and $f(0) = 0$. There is another zero, not necessarily distinct, say, $f(c_1)=0$. Since $f(1) = \mathcal{E}_\sigma[\psi_\infty]  <0$, necessarily $c_1<1$. This means $f(c)$ decreases further when $c$ increases above $1$.   
%
%
  Now, if $\|\psi_\infty\|<1$, then there is a $c> 1$, such that $\|\sqrt{c} \psi_\infty\|=1$, but then $\mathcal{E}_\sigma [\sqrt{c}\psi_\infty] < \mathcal{E}_{\sigma, 0}$, which is a contradiction. Therefore,  $\|\psi_\infty\|=1$.  
\end{proof}

\subsection{The Euler-Lagrange equation for the tempered energy functional and locality of minimizers}

 It is easily seen, see Subsection \ref{Appendix_EL}, that the energy the Euler-Lagrange equation for functional (\ref{tempered_BH}) assumes the form: 
\begin{equation}\label{tempered_EL}
  \left[\frac{U}{2} Q(n) - \left( \frac{U}{2} + \mu\right)\Omega(n)\right] n^{-2\sigma} x_n - t\,n^{-2\sigma} \sum_l  x_{l}^*  \sum_{d|ln} x_d x_{ln/d} \, l^{-2\sigma}
= \lambda\, x_n \quad \mbox{ for all } n.
\end{equation}
We have established that there are solutions for such a system of equations for any $\sigma >0$. Since $\sigma$ can be arbitrarily small, we take these to be approximate solutions of the corresponding equations with $\sigma =0$.   
Note that, by Cauchy-Schwartz 
\[
\sum_l  x_{l}^*  \sum_{d|ln} x_d x_{ln/d} \, l^{-2\sigma} \leq
\left(\sum_l |x_l|^2l^{-2\sigma}\right)^{1/2} \left(\sum_l \left|\sum_{d|ln} x_d x_{ln/d}\right|^2 \, (ln)^{-2\sigma}\right)^{1/2} n^\sigma \leq
\|\psi\|_{-\sigma}\, \|\psi \odot \psi\|_{-\sigma} \,n^\sigma
\]
This shows that if $\lambda\neq 0$, then $x_n$ diminish to zero at least as fast as $n^{-\sigma}$, which is a weak form of localization. 
\vspace{.5cm}

\noindent
\emph{Remark.} 
In the special case of one-particle space, only finitely many $x_n$ ($n$ prime) are non-zero. In this case,
\[
\sum_{d|ln} x_d x_{ln/d} =\left\{\begin{array}{ccc}
                                  2x_lx_n & \mbox{ if } & l\neq n \\
                                   &  &  \\
                                  x_n^2 & \mbox{ if } & l=n 
                                \end{array}
 \right.
\] 
Indeed, by assumption, $l,n$ and $d$ are primes. This implies
\[
\sum_l  x_{l}^*  \sum_{d|ln} x_d x_{ln/d} \, l^{-2\sigma} = \left(2\|\psi\|^2_{-\sigma} - |x_n|^2 n^{-2\sigma}\right) x_n.  
\]
Therefore, (\ref{tempered_EL}) assumes a simpler form:
\begin{equation}\label{tempered_EL_primes}
\left(-\mu n^{-2\sigma} - t\,n^{-2\sigma} \left(2\|\psi\|^2_{-\sigma} - |x_n|^2 n^{-2\sigma}\right) -
 \lambda\right)\, x_n =0 \quad \mbox{ for all  primes } n.
\end{equation}
Assume that $x_n\neq 0$ for infinitely many primes $n$. Then necessarily  $\lambda = 0$ as all other terms in the parenthesis tend to zero as $n$ tends to infinity. In such a case, we would have
\[
|x_n|^2 = \left(\mu +2 t \|\psi\|^2_{-\sigma}\right) n^{2\sigma},
\]
but the right hand side tends to infinity, which is a contradiction. Therefore only finitely many $x_n$ can be nonzero. This is consistent with the case $\sigma =0$ resolved in Subsection \ref{Subsect_finite_sols}. 
We do not know if a similar statement could be made for all $N>1$. In other words, \emph{it is an open problem whether or not the minimizers can be infinitely supported.    }

\section{Summary}
We have introduced a type of quantum dynamics of bosons which involves a nonlocal nonlinearity. We have analyzed its relation to the core quantum theory, which revolves around the canonical and generalized bosons. We have also provided examples of solutions of this dynamic, using the methods of harmonic analysis specifically for the group of fractions. We have demonstrated existence of approximate solutions in the finite-particle spaces. Toward this end, we have analysed certain nonlinear functionals constructed with the use of adapted Sobolev spaces. We emphasize the divergence of this framework from the known geometric theories. Here, there is no underlying manifold. Instead the central role is played by $\mathbb{T}^\omega$, i.e. the dual of the multiplicative group of fractions.


\section{Appendix} \label{Section_Appendix}

\subsection{Harmonic analysis on the multiplicative group of positive rationals } \label{Appendix_FT_Q_plus}

It was demonstrated in \cite{SF22} that the creation and annihilation operators (\ref{on_arithmetic}), together with the Fock space in its $\ell_2(\mathbb{N})$ representation, admit a type of Fourier-dual representation. Below, we briefly outline harmonic analysis on the group of positive rationals, which furnishes this dual picture. The general theory of abstract harmonic analysis is presented in standard textbook expositions such as \cite{Rudin} and \cite{Folland}. However, such general treatments tend to obscure some of the important idiosyncrasies of specific models. Harmonic analysis on $\mathbb{Q}_+$ has been used in Number Theory \cite{Elliott} but, to our knowledge, not in quantum theory prior to the aforementioned article.

First note the isomorphism of abelian groups:
\begin{equation}\label{rationals}
 \mathbb{Q}_+ \equiv \bigoplus\limits_{p: \in \mathcal{P}} \mathbb{Z} \quad \mbox{ given by the prime factorization  } \quad \mathbb{Q}_+ \ni w = \prod_{p\in \mathcal{P}} p^{a_p}, \, a_p \in\mathbb{Z}.
\end{equation}
In light of this, the dual group of $\mathbb{Q}_+$ is
\[
\mathbb{T}^\omega= \prod_{p \in \mathcal{P}} U(1).
\]
 It is often  useful to identify $\mathbb{T}^\omega$ with the set of completely multiplicative functions $\chi: \mathbb{Q}_+\rightarrow U(1)$, i.e. functions that satisfy $\chi(uw)= \chi(u)\chi(w)$. Namely,
\begin{equation}\label{mult_chi}
\prod_{p \in P} U(1) \ni (\theta_2, \theta_3, \theta_5, \ldots ) \,\, \mbox{ corresponds to } \chi \mbox{ characterized by }\,\, \chi(p) = p^{i \theta_p} , p \in \mathcal{P}.
\end{equation}
It is also useful to introduce a change of variable setting $ \mu_p = \theta_p\,\log p/(2\pi)$,
so that
\[
\frac{\log p}{2\pi}\int\limits_{0}^{2\pi/\log p } p^{ik\theta_p} \, d\theta_p =
\int\limits_{0}^{1} e^{2\pi ik\mu_p} \, d\mu_p = \left\{\begin{array}{cc}
                                                         1 & k = 0 \\
                                                         0 & k \neq 0
                                                       \end{array}\right. .
\]
When equipped with the product topology $\mathbb{T}^\omega$ is, by virtue of the Tychonoff Theorem, a compact space. Secondly, it admits a unique Borrel measure $d\vec{\mu}$, which satisfies
\[
d\vec{\mu} \left( (\sigma_2, \beta_2]  \times  \ldots \times (\sigma_p, \beta_p] \times (0,1]  \times (0,1] \times \ldots  \right) =
|\beta_2 - \sigma_2| \ldots |\beta_p - \sigma_p|.
\]
In particular, $d\vec{\mu} ( \mathbb{T}^\omega) =1 $, i.e. the measure is probabilistic. At the same time, $\mathbb{Q}_+$ itself is equipped with the discrete (counting) measure.

Utilizing these identifications---identifying $\chi$ with  $(\theta_2, \theta_3, \theta_5, \ldots )$ and with $(\mu_2, \mu_3, \mu_5, \ldots )$----the Fourier transform appears in several different guises. First, it is defined via:
 \begin{equation}\label{FT_def1}
   \hat{f}(\chi) = \sum_{w\in \mathbb{Q}_+} f(w)\,\chi(w)^* \mbox{ where } f:\mathbb{Q}_+ \rightarrow \mathbb{C}.
 \end{equation}
This defines $\hat{f}: \mathbb{T}^\omega\rightarrow \mathbb{C}$ with $\chi$ as its argument. The inverse transform is then given by
\begin{equation}\label{FT_Inv_def1}
  f(w) = \int \hat{f}(\chi)\, \chi(w)\, d\vec{\mu}(\chi) = \int\limits_{0}^{1} \int\limits_{0}^{1} \int\limits_{0}^{1} \ldots \hat{f}(\chi)\, \chi(w) \, d\mu_2\, d\mu_3\, d\mu_5\ldots
\end{equation}
This  notation is useful in particular when expressing the following fundamental properties: First, for a fixed arbitrary $u \in \mathbb{Q}_+ $, a direct calculation shows
\begin{equation}\label{mult}
 \mbox{Let } g(w): = f(uw) \mbox{ for all } w \in \mathbb{Q}_+.  \mbox{ Then } \hat{g}(\chi) = \chi(u^{-1}) \hat{f}(\chi).
\end{equation}
Second, note that the measure $d\mu(\chi)$ is invariant with regards to circular shifts along the $U(1)$ components. Thus, for a fixed collection $\vec{\nu} = (\nu_2, \nu_3, \nu_5, \ldots)$ if we define $\hat{g}$ via
\[
\hat{g}(\vec{\mu} ) : = \hat{f}(\vec{\mu} - \vec{\nu} ),
\]
then $\|\hat{f}\| = \| \hat{g} \|$ since $d\vec{\mu}$ is shift invariant. Calculating the inverse transform we readily obtain
\begin{equation}\label{shift}
  g(w) = \chi_{\vec{\nu}}(w) \, f(w), \quad \mbox{ where } \quad \chi_{\vec{\nu}} \equiv (\nu_2, \nu_3, \nu_5, \ldots).
\end{equation}

We now turn attention to the Hilbert-space theoretic aspects. In some ways it is more akin to that of $L_2(\mathbb{R})$, rather than $L_2(U(1))$. That is because the characters are not square integrable functions, indeed:
\[
\|\chi\|^2 =  \sum_{w\in \mathbb{Q}_+} \chi(w)\, \chi(w)^* =  \sum_{w\in \mathbb{Q}_+} 1 = \infty .
\]
We will briefly discuss the Parseval identity. On one hand, it follows from the general Pontryagin Theorem. On the other hand, it is instructive to observe it directly in this context. The calculation becomes more explicit with the use of (\ref{rationals}) to identify rational numbers with finitely supported sequences of integers:
\[
\vec{w} = (w_2, w_3, w_5, \ldots) \in \bigoplus\limits_{p: \in P} \mathbb{Z}.
\]
Accordingly, we use the notation:
$
\vec{a}\cdot \vec{\mu} = w_2\mu_2 + w_3\mu_3 + w_5\mu_5 + \ldots \in \mathbb{C}.
$
We can express the Fourier transform of the point measure:
\[
\delta_{\vec{w}} (x ) \mapsto e^{2\pi i \vec{w}\cdot \vec{\mu}} ,
\]
so that $x$ and $\hat{\mu}$ are dual variables.
Thus, for an arbitrary function $f$ on $\mathbb{Q}_+$ we have
\begin{equation}\label{FT_def2}
  f(w) = \sum_{\vec{w}} f(\vec{w}) \, \delta_{\vec{a}} (w ) \quad \mapsto \quad \hat{f}(\vec{\mu}) = \sum_{\vec{w}} f(\vec{w}) \, e^{2\pi i \vec{w}\cdot \vec{\mu}} ,
\end{equation}
where the summation is over all $\vec{w} \in \bigoplus\limits_{p: \in P} \mathbb{Z}$. It is easy to see that definitions (\ref{FT_def1}) and (\ref{FT_def2}) are equivalent.
Now, suppose $f: \mathbb{Q}_+ \rightarrow \mathbb{C}$ is square summable, i.e.
\begin{equation}\label{l2_Q_+}
\|f\|^2 = \sum_{w\in \mathbb{Q}_+} |f(w)|^2 = \sum\limits_{\vec{w}}\,
| f(\vec{w})|^2 < \infty .
\end{equation}
Note that
\begin{equation}\label{intexp}
  \int d\vec{\mu} \,\, e^{2\pi i \vec{w}\cdot \vec{\mu}} = \left\{ \begin{array}{cc}
                                                                  1 & \vec{w} = 0 \\
                                                                  0 & \mbox{ oth.}
                                                                \end{array}
  \right.
\end{equation}
It is easily seen that  $\|\hat{f}\|  = \|f\|$ (Parseval identity).
This is equivalent to stating that as $\vec{a}$ runs over all finitely supported sequences of integers $\delta_{\vec{a}}(x)$ furnish an orthonormal basis in $\ell_2(\mathbb{Q}_+)$ while $\exp(2\pi i \vec{a}\cdot \vec{\mu})$ furnish such a basis in $L_2(\mathbb{T}^\omega, d\vec{\mu})$, so that the map (\ref{FT_def2}) is unitary.

We denote by $H_2(\mathbb{T}^\omega, d\vec{\mu})$ the subspace of functions whose only nonzero coefficients are those whose index is a natural number, i.e., the coefficients with fractional indices vanish. Of great importance in what is to follow is the natural orthogonal projection:
\begin{equation}\label{Pplus}
  P_+: L_2(\mathbb{T}^\omega, d\vec{\mu}) \rightarrow H_2(\mathbb{T}^\omega, d\vec{\mu}).
\end{equation}
This is unitarily equivalent to the projection (denoted by the same symbol) $P_+: \ell_2(\mathbb{Q}_+) \rightarrow \ell_2(\mathbb{N})$.
 Note that the Fourier transform introduced in the previous section establishes an equivalence of the spaces and subspaces as follows:
\begin{equation}\label{diagram}
\begin{array}{ccc}
  \bigotimes\limits_{p\in\mathcal{P}} H_2(U(1)) & \subset & \bigotimes\limits_{p\in\mathcal{P}} L_2(U(1))\\
  \updownarrow\mbox{FT} & & \updownarrow\mbox{FT} \\
  \ell_2(\mathbb{N}) & \subset & \ell_2(\mathbb{Q}_+)
\end{array}
\end{equation}
Even though the arithmetic model of the Fock space corresponds to the Hecke subspace, the group-duality based theory requires that in order to understand the whole picture we cannot loose the sight of the entire $ \ell_2(\mathbb{Q}_+)$.

\subsection{Remarks on the functional $\ket{\psi} \mapsto\|\psi \odot \psi\|^2$: the case of $\ket{\psi}$ supported on a pairwise co-prime set}\label{Appendix_coprime}

We say that a set of positive integers $S\subset \mathbb{N}$ is a pairwise coprime if any two distinct $m,n \in S$ are coprime, i.e., $\mbox{gcd}(m,n) =1$. We now make an Ansatz:
\begin{equation}\label{Ansatz_coprime}
 \ket{\psi} = \sum_{n \in S} x_n \ket{n},\quad x_n \neq 0 \mbox { for all } n\in S (\mbox {where } S \mbox{ is pairwise coprime})
\end{equation}
In what follows we will use the radial variables $x_n = r_ne^{i\theta_n}$, $r_n>0$. We also assume the constraint
\begin{equation}\label{normal_coprime}
  \|\psi\|^2 = \sum_{n\in S} r_n^2 =1.
\end{equation}
As it turns out, this assumption has the effect of simplifying the functional $\ket{\psi} \mapsto\|\psi \odot \psi\|^2$ and, particularly, the determination of its critical points. Indeed, in this case, 
\begin{equation}\label{Veronese}
  \ket{\psi} \odot \ket{\psi} = \sum_{n} x_n^2 \,\ket{n^2}\, +\, 2\sum_{(m,n): m < n} x_m x_n \ket{mn}. 
\end{equation}
Note that all the Fock basis vectors in this formula are mutually distinct. (It is interesting to note that in this case the map $\ket{\psi} \mapsto \ket{\psi} \odot \ket{\psi}$ is akin to the classical Veronese map.)  We have
\begin{equation}\label{F_Veron}
   F =\|\psi \odot \psi\|^2 =  \sum_{n} r_n^4  + 4 \sum_{(m,n): m< n} r_m^2 r_n^2. 
\end{equation}
In particular, in this case $F$ does not depend on the phases $\theta_n$. We will now summarize its properties as a function of real variables $r_n$.  
Note that $F$ is the sum off all entries in the matrix, say $M$, whose rows and colums are indexed by $n\in S$, such that $M_{n,n} = r_n^4$ and $M_{k,l} = 2r_k^2r_l^2$ whenever $k\neq l$. Evaluating the sum, say, column by column, and using the constraint (\ref{normal_coprime}), one readily finds: 
\begin{equation}\label{F_Veron_simp}
   F =  2 - \sum_{n} r_n^4. 
\end{equation}
It is then easily seen that $1<F\leq 2-1/|S|$, where $|S|$ is the number of elements in $S$. Note also that the maximum can only be attained  if  $|S|$ is finite. When it is infinite, $F$ assumes values arbitrarily close to the value of $2$, but cannot attain it. The maximum occurs at the only critical point of $F$, subject to (\ref{normal_coprime}) and the condition $r_n>0$ for all $n\in S$. It occurs when all $r_n$ are equal, i.e.,  $r_n = 1/\sqrt{|S|}$ for all $n\in S$, . The minimum of $F$ is attained on the boundary of the region, when all but one coefficients turn to zero.

\subsection{Remarks on the functional $\ket{\psi} \mapsto\|\psi \odot \psi\|^2$: the case of two particle states supported on two sites}\label{subsection_three_on2}

In light of the previous subsection, it is interesting to address that case when (\ref{Ansatz_coprime}) does not hold, i.e., Fock basis vectors need not be indexed by mutually co-prime integers. As it turns out, in such a case the functional in question displays interesting dependence on the phase variables $\theta_n$, and its critical points impose a constraint on the phases. 
Here, we analyse the simplest nontrivial case which is handled via calculus. Namely, we consider a state $\ket{\psi} = x_0 \ket{p^2} + x_1\ket{pq} + x_2 \ket{q^2}$, where $p,q$ are distinct primes. Again, we will use polar variables $x_n = r_n e^{i\theta_n}$. Assume the normalization
\begin{equation}\label{normal}
  |x_0|^2 + |x_1|^2 + |x_2|^2 =1.
\end{equation}
We are primarily interested in the case in which $r_0,r_1,r_2>0$. 
Note that 
\begin{equation}\label{psi_cdot_psi}
  \ket{\psi} \odot \ket{\psi} = x_0^2 \ket{p^4} + 2x_0x_1 \ket{p^3q} + (2x_0x_2 + x_1^2)\ket{p^2q^2} + 2x_1x_2 \ket{q^3p} + x_2^2 \ket{q^4}.
\end{equation}
We will analyse the critical points (CPs) of the function
\begin{equation}\label{theF_first}
     \begin{array}{cccl}
                  F  = F(x_0, x_1, x_2)   = &               \|\ket{\psi} \odot \ket{\psi} \|^2  & = &|x_0|^4 + 4|x_1|^2 \left(|x_0|^2 + |x_2|^2 \right)+ |x_2|^4 + |2x_0x_2 +x_1^2|^2 \\
                                &   &  &  \\
                                 &  & = & r_0^4 + 4r_1^2 \left(r_0^2 + r_2^2 \right)+ r_2^4 + |2x_0x_2 +x_1^2|^2. \\
                                \end{array}
\end{equation}
Note that the action of the circle group via the uniform phase shift: 
\begin{equation}\label{phase_shift}
  (x_0,x_1,x_2) \mapsto e^{i\phi}(x_0,x_1,x_2)
\end{equation}
leaves $F$ invariant. It is easily seen  by direct analysis that
\begin{equation}\label{Fgr1}
  F(x_0, x_1, x_2) \geq 1\quad \mbox{ provided } (\ref{normal}) \mbox{ holds. } 
\end{equation}  
\vspace{.5cm}

\noindent
\textbf{A direct argument to find the maximum:} Since $F$ is a smooth function, it attains a maximum on the 5-sphere defined by (\ref{normal}). Let the point $(r_0\, e^{i\theta_0}, r_1\, e^{i\theta_1},r_2\, e^{i\theta_2} )$ be a point where $F$ attains its maximum. Only the last term in (\ref{theF_first}) depends on the phases. Keeping the amplitudes $r_0,r_1,r_2$ fixed, this term is maximized when $2x_0x_2$ and $x_1^2$ have equal phases, i.e.,
\begin{equation}\label{phase_constraint_2part_2sites}
   2\theta_1 = \theta_0 + \theta_2\, (\mbox{mod } 2\pi).  
\end{equation}
Therefore, at the maximum 
\[
|2x_0x_2 + x_1^2| = 2r_0 r_2 + r_1^2 = 1 - (r_0-r_2)^2.
\] 
This expression is maximized when $r_0 = r_2$. 
In addition, on each concentric quarter-circles $r_0^2+r_2^2 = r^2$, the expression $$r_0^4 + 4 \left(1 - r_0^2 - r_2^2 \right)\left(r_0^2 + r_2^2 \right)+ r_2^4$$ attains its maximum when $r_0 = r_2$. It follows that $F$ attains its maximum on the line  $r_0=r_2$, where it assumes the form:
\[
F = 2r_0^4 + 4(1-2r_0^2)2r_0^2 +1.
\]
This function assumes its maximum when $r_0 = \sqrt{2/7}$. This identifies the maximum as a unique point up to the uniform phase shift phase, namely:
\[
e^{i\phi}\,\left(\sqrt{\frac{2}{7}}\, e^{i\theta_0}, \sqrt{\frac{3}{7}}\,e^{i(\theta_0+\theta_2)/2}, \sqrt{\frac{2}{7}}\, e^{i\theta_2}\right).
\] 

\vspace{.5cm}

\noindent
\textbf{A detailed pedestrian analysis via calculus:}
Developing the last term in (\ref{theF_first}), we obtain:
\begin{equation}\label{theF}
    F = |x_0|^4 + |x_1|^4+ |x_2|^4 + 4\left(|x_0|^2 |x_1|^2 + |x_1|^2|x_2|^2 + |x_2|^2|x_0|^2 \right) + 2(x_0x_1^{*2}x_2 + x_0^*x_1^{2}x_2^*),
\end{equation}
subject to the constraint (\ref{normal}). The problem of localizing the CPs is simplified by modifying the function $F$ as follows:
\begin{equation}\label{theF_modified}
    F =1 + 2\left(|x_0|^2 |x_1|^2 + |x_1|^2|x_2|^2 + |x_2|^2|x_0|^2 \right) + 2(x_0x_1^{*2}x_2 + x_0^*x_1^{2}x_2^*).
\end{equation}
Indeed, the functions defined in (\ref{theF}) and (\ref{theF_modified}) are identical on the 5-sphere defined by the constraint. (We point out there are more equivalent forms of the function $F$ on the 5-sphere.) The conditions for $(x_0, x_1, x_2)$ to be a CP are found by differentiating with respect to the variables $x_j^*$, namely:
\begin{equation}\label{array_1}
\left\{\begin{array}{lcc}
  (|x_1|^2 + |x_2|^2)x_0 + x_1^2x_2^* &=& \lambda x_0 \\
  &&\\
  (|x_0|^2 + |x_2|^2)x_1 + 2 x_0 x_1^*x_2 &=& \lambda x_1 \\
  && \\
 (|x_1|^2 + |x_0|^2)x_2 + x_1^2x_0^* &=& \lambda x_2 
 \end{array}\right.
\end{equation}
Note that (\ref{array_1}) is equivalent to:
 \begin{equation}\label{array_2}
\left\{\begin{array}{lcc}
  (r_1^2 + r_2^2)r_0 + r_1^2r_2\, e^{i(2\theta_1 - \theta_0 - \theta_2)} &=& \lambda r_0 \\
  &&\\
  (r_0^2 + r_2^2)r_1 + 2 r_0 r_1 r_2\, e^{i(\theta_0 + \theta_2 - 2\theta_1)}&=& \lambda r_1 \\
  &&\\
 (r_1^2 + r_0^2)r_2 + r_1^2r_0\, e^{i(2\theta_1 - \theta_0 - \theta_2)}&=& \lambda r_2 
 \end{array}\right.
\end{equation}
As an immediate conclusion, we obtain a constraint on the phases:
\begin{equation}\label{phase_constraint}
  \mbox{ either }\,\, \theta_0 + \theta_2 - 2\theta_1 = 0\,  \mod 2\pi\,\, (\mbox{ Case } +) ,\quad \mbox{ or }\quad \theta_0 + \theta_2 - 2\theta_1 = \pi\, \mod 2\pi\,\, (\mbox{ Case } -) . 
\end{equation}
In either one of these cases the system of equations (\ref{array_2}) reduces to a system of three equations with three real unknowns:
 \begin{equation}\label{array_3}
\left\{\begin{array}{lcc}
  (r_1^2 + r_2^2)r_0 \pm r_1^2r_2 &=& \lambda r_0 \\
  &&\\
  (r_0^2 + r_2^2)r_1 \pm 2 r_0 r_1 r_2 &=& \lambda r_1 \\
  &&\\
 (r_1^2 + r_0^2)r_2 \pm r_1^2r_0 &=& \lambda r_2 
 \end{array}\right.
\end{equation}
Next, we note that (\ref{array_3}) is the condition for CPs for functions
\begin{equation}\label{theF_pm}
    F_\pm = r_1^2 (r_0 \pm r_2)^2 + r_0^2r_2^2,
\end{equation}
on the section of the sphere defined via $r_0^2 + r_1^2 + r_2^2 =1 $, $r_0, r_1, r_2 >0$.
Note that functions $F_\pm$ can be represented by 
\begin{equation}\label{theF_pm_2vars}
    F_\pm(r_0,r_2) = (1- r_0^2 -r_2^2) (r_0 \pm r_2)^2 + r_0^2r_2^2,\quad r_0, r_2> 0,\,\, r_0^2+ r_2^2 < 1.
\end{equation}
 The CPs of $F_\pm(r_0,r_2)$ are characterised by the equations for the vanishing partial derivative of these functions, i.e.,
 \begin{equation}\label{array_4}
\left\{\begin{array}{lcc}
  -(r_0 \pm r_2)^2 r_0 + (1- r_0^2 -r_2^2)(r_0 \pm r_2) + r_0r_2^2 &=& 0 \\
  &&\\
  -(r_0 \pm r_2)^2 r_2 \pm (1- r_0^2 -r_2^2)(r_0 \pm r_2) + r_0^2r_2 &=& 0
 \end{array}\right.
\end{equation}

It is straightforward that in the Case $+$ we have only one solution; namely, $r_0 = r_2 = \sqrt{\frac{2}{7}}$. Moreover, the Hessian at this critical point is
\[
\nabla^2 F_+ = -\frac{2}{7}
\left(
  \begin{array}{cc}
    19 & 9 \\
    9 & 19 \\
  \end{array}
\right),
\]
which has two negative eigenvalues, indicating that this CP is a local maximum. 

In the Case $-$, examining the sum and the difference of the two equations, leads to
 \begin{equation}\label{sum_diff}
\left\{\begin{array}{lcc}
  r_0^2 - 3r_0r_2+  r_2^2 &=& 0 \\
  &&\\
   3 r_0^2 - r_0r_2+  3r_2^2 &=& 2,
 \end{array}\right.
\end{equation} 
so that one can see $r_0r_2 = 1/4$, and $r_0^2 + r_2^2 = 3/4$. The solutions are $(r_0,r_2) = (\frac{\sqrt{5}+1}{4}, \frac{\sqrt{5}-1}{4} )$ and $(r_0,r_2) = (\frac{\sqrt{5}-1}{4}, \frac{\sqrt{5}+1}{4} )$.
The Hessian at the first point is
\[
\nabla^2 F_- = \frac{1}{4}
\left(
\begin{array}{cc} 
-5\,\sqrt{5}-1 & 6 \\ 
6 & 5\,\sqrt{5}-1 
\end{array}\right),
\]
which has one positive and one negative eigenvalue. Due to symmetry between these two CPs, both are saddle points.
\vspace{.5cm}

\noindent
\emph{Remarks.} It is easy to find solutions of the system (\ref{array_1}) which lie on the boundary of the region of interest: These are:
 \begin{equation}\label{sols_simp}
 (e^{i\theta_0},0,0),\quad  (0,e^{i\theta_1},0),\quad (0,0,e^{i\theta_2}), \mbox{ and } \left(\frac{e^{i\theta_0}}{\sqrt{2}}, 0, \frac{e^{i\theta_2}}{\sqrt{2}}\right),
 \end{equation}
 where $\theta_0, \theta_1,\theta_2$ are arbitrary. Since $F(1,0,0)=F(0,1,0)=F(0,0,1)=1 $, the first three CPs are minima. Note that the condition $|x_1| = 0$ determines a submanifold of the 5-sphere, a 3-sphere. The point $\left(\frac{e^{i\theta_0}}{\sqrt{2}}, 0, \frac{e^{i\theta_2}}{\sqrt{2}}\right)$ may be considered as a CP point of $F$ restricted to the 3-sphere. This solution is within the scope of the cases considered in Subsection \ref{Appendix_coprime}. As was established this point is a maximum in the 3-sphere. However, its behavior is more complex in the 5-sphere. Namely, $F_+(r, \sqrt{1-2r^2},r)$ decreases while  $F_-(r, \sqrt{1-2r^2},r)$ increases as $r$ approaches $1/\sqrt{2}$ from the left, see Fig, \ref{Fig_FpFm}. This means that the behavior of $F$ near this point is sensitive to the phases at $x_0,x_1,x_2$. In other words, it is an anisotropic CP.   
We summarize the CPs for $F$ in the case of two-particle states supported in two locations:

\begin{equation}\label{eig_Umu_zero}
'\left(\begin{array}{ccccll} E & r_0 & r_1 & r_2 & \mbox{ phases} & \mbox{ type } \\ 
&&&&\\
1 & 1 & 0 & 0 & \mbox{ arbitrary }& \mbox{ min }\\
&&&&\\
 1 & 0 & 1 & 0 &\mbox{ arbitrary }& \mbox{ min }\\ 
 &&&&\\
 1 & 0 & 0 & 1 &\mbox{ arbitrary }& \mbox{ min }\\
  &&&&\\
 \frac{3}{2} & \frac{\sqrt{2}}{2} & 0 & \frac{\sqrt{2}}{2} &\mbox{ arbitrary }& \mbox{ anisotropic}\\ 
 &&&&\\
  \frac{15}{7} & \sqrt{\frac{2}{7}} & \sqrt{\frac{3}{7}}& \sqrt{\frac{2}{7}} & 2\theta_1 = \theta_0 + \theta_2 & \mbox{ max } \\
  &&&& \\
  \frac{5}{4} & \frac{\sqrt{5}+1}{4} &\frac{1}{2} & \frac{\sqrt{5}-1}{4} &  2\theta_1 = \theta_0 + \theta_2 + \pi& \mbox{ saddle }\\
  &&&&\\
    \frac{5}{4} & \frac{\sqrt{5}-1}{4} &\frac{1}{2} & \frac{\sqrt{5}+1}{4}  & 2\theta_1 = \theta_0 + \theta_2+\pi & \mbox{ saddle }
  \end{array}\right)
\end{equation}

 \begin{figure}[ht!]
\centering
\includegraphics[width=120mm]{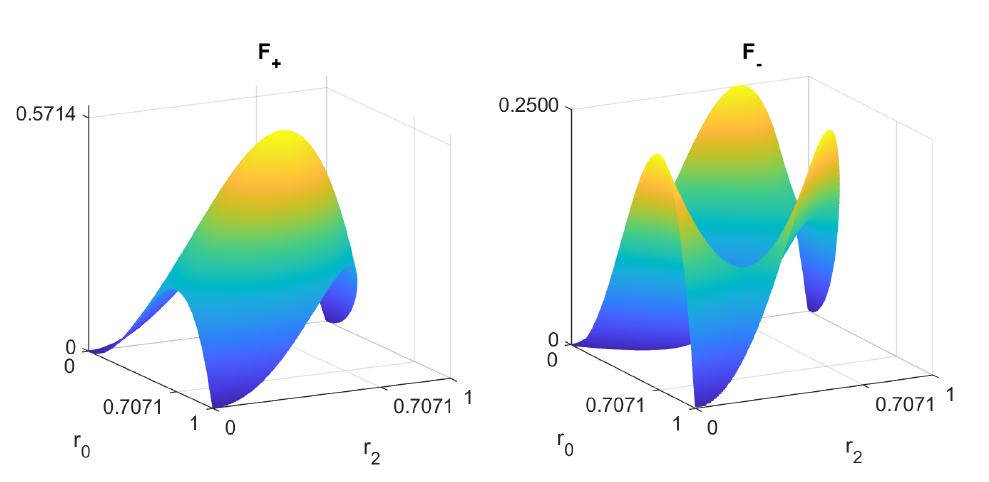}
\caption{Functions $F_+ = F_+(r_0,r_2)$ and $F_- = F_-(r_0,r_2)$.  Both graphs are symmetric with respect to the interchange $r_0 \leftrightarrow  r_2$. $F_+$ attains the maximum where $r_0 = r_2 = \sqrt{14}/7$. $F_-$ has saddle points at $(\frac{\sqrt{5} +1}{4}, \frac{\sqrt{5} -1}{4})$  and at $(\frac{\sqrt{5} -1}{4}, \frac{\sqrt{5} +1}{4})$.   Note the difference of behavior at the singular point on the boundary where $r_0 = r_2 = 1/\sqrt{2}$.
}
\label{Fig_FpFm}
\end{figure}

\subsection{The Euler-Lagrange equation in tempered norms.} \label{Appendix_EL}

We need to determine the form of the formal adjoint of a matrix operator, denoted $A$, with respect to the tempered (i.e., $\langle \, \, | \,\, \rangle_{-\sigma}$) norm. Let $\ket{\psi} = \sum_{n \in S} x_n \ket{n}$ and $\ket{\phi} = \sum_{n \in S} y_n \ket{n}$. We have
\begin{eqnarray*}
 \langle \psi \,|\, A |\phi \rangle_{-\sigma}  &=& \sum_{n} x_n^*\, \sum_{k} A_{nk}y_k\, n^{-2\sigma} \\
   &&  \\
   &=& \sum_{k} \, \left(\sum_{n} A_{nk}^* n^{-2\sigma} x_n \right)^* y_k  \\
   &&  \\
   &=& \sum_{k} \, \left(\sum_{n} k^{2\sigma} A_{nk}^* n^{-2\sigma} x_n \right)^* y_k\, k^{-2\sigma} \\
   && \\
   &=&  \langle \phi \, | A_\sigma^\dagger|\psi \rangle_{-\sigma}^*, \quad \mbox{ where } \quad (A_\sigma^\dagger)_{kn} = k^{2\sigma} A_{nk}^* n^{-2\sigma}.
\end{eqnarray*}
In particular, a diagonal operator with real entries is formally self-adjoint regardless of the value of $\sigma$.

Let $X$ be a matrix operator determined by $\ket{\psi}$ as in (\ref{theX}). The Euler-Lagrange (E-L) equation for the functional $\ket{\psi} \mapsto 4^{-1}\|\ket{\psi} \odot \ket{\psi}\|_{-\sigma}^2$ has the form $ X_\sigma^\dagger X \ket{\psi} =\lambda \, \ket\psi$. Let us take a closer look at this condition. First, let $\ket\phi = X\ket\psi$, so that
\[
y_n = \sum_{d|n} x_d x_{n/d}.
\]
Let us evaluate the $k$-th coordinate of $X_\sigma^\dagger \ket\phi $, i.e., 
\begin{eqnarray*}
\langle \delta_k\, |\, X_\sigma^\dagger \ket\phi  &=& \sum_{n} (X_\sigma^\dagger)_{kn}\,y_n 
=k^{2\sigma} \sum_n  X_{nk}^* n^{-2\sigma} y_n \\
&& \\
&=& k^{2\sigma} \sum_n  x_{n/k}^* n^{-2\sigma} y_n
=  k^{2\sigma} \sum_l  x_{l}^* (lk)^{-2\sigma} y_{lk} \\
&& \\
&=& 
\sum_l  x_{l}^*  \sum_{d|lk} x_d x_{lk/d} \, l^{-2\sigma} 
\end{eqnarray*}
Also,
\[
\langle \delta_k\, |\, X_\sigma^\dagger \ket\phi_{-\sigma}  =
k^{-2\sigma}\sum_l  x_{l}^*  \sum_{d|lk} x_d x_{lk/d} \, l^{-2\sigma}. 
\]
 Thus, the energy functional (\ref{tempered_BH}) results in the E-L equation in the form (\ref{tempered_EL}).

\end{document}